\documentclass[english]{article}
\usepackage[T1]{fontenc}
\usepackage[latin9]{inputenc}
\usepackage{geometry}
\usepackage{amsthm}
\usepackage{amsmath}
\usepackage{amssymb}
\usepackage{color}

\makeatletter
\numberwithin{equation}{section}
\numberwithin{figure}{section}
\theoremstyle{plain}

  \theoremstyle{definition}
  
  \theoremstyle{plain}
  
  \theoremstyle{definition}
  
  \theoremstyle{remark}
  
  \theoremstyle{remark}
  \newtheorem*{rem*}{\protect\remarkname}
  \theoremstyle{plain}

\@ifundefined{date}{}{\date{}}

\input{Qcircuit}
\usepackage{graphicx,amsmath, amsthm, amssymb,color,url,booktabs,comment}  
\usepackage[colorlinks]{hyperref} 

\newcommand{\nc}{\newcommand}
\nc{\rnc}{\renewcommand}
\newcommand{\proj}[1]{|#1\rangle\langle #1|}

\nc{\vev}[1]{\langle#1\rangle}
\nc{\grad}{{\vec{\nabla}}}

\DeclareMathOperator{\tr}{\mbox{\rm Tr}}

\newcommand{\be}{\begin{equation}}
\newcommand{\ee}{\end{equation}}
\newcommand{\bea}{\begin{eqnarray}}
\newcommand{\eea}{\end{eqnarray}}
\newcommand{\nn}{\nonumber}
\newcommand{\bi}{\begin{itemize}}
\newcommand{\ei}{\end{itemize}}
\newcommand{\bn}{\begin{enumerate}}
\newcommand{\en}{\end{enumerate}}
\def\beas#1\eeas{\begin{eqnarray*}#1\end{eqnarray*}}
\def\ba#1\ea{\begin{align}#1\end{align}}
\nc{\bas}{\[\begin{aligned}}
\nc{\eas}{\end{aligned}\]}
\nc{\bpm}{\begin{pmatrix}}
\nc{\epm}{\end{pmatrix}}

\def\nn{\nonumber}

\def\L{\left} 
\def\R{\right}
\def\ra{\rightarrow}

\newtheorem{dfn}{Definition}

\makeatletter
\newtheorem*{rep@theorem}{\rep@title}
\newcommand{\newreptheorem}[2]{%
\newenvironment{rep#1}[1]{%
 \def\rep@title{#2 \ref{##1}}%
 \begin{rep@theorem}}%
 {\end{rep@theorem}}}
\makeatother

\newreptheorem{thm}{Theorem}
\newreptheorem{lem}{Lemma}

\def\eps{\epsilon}

\def\benum{\begin{enumerate}}
\def\eenum{\end{enumerate}}
\def\bit{\begin{itemize}}
\def\eit{\end{itemize}}
\def\bdesc{\begin{description}}
\def\edesc{\end{description}}

\nc{\todo}[1]{\textcolor{red}{todo: #1}}

\def\begsub#1#2\endsub{\begin{subequations}\label{eq:#1}\begin{align}#2\end{align}\end{subequations}}
\nc\qand{\qquad\text{and}\qquad}
\nc\mnb[1]{\medskip\noindent{\bf #1}}
\nc\mn{\medskip\noindent}

\newtheorem{theorem}{Theorem}
\newtheorem{proposition}[theorem]{Proposition}
\newtheorem{lemma}[theorem]{Lemma}

\newtheorem{corollary}[theorem]{Corollary}
\newtheorem{definition}[theorem]{Definition}

\newenvironment{step}
  {
    \begin{enumerate}

  }
  {\end{enumerate}}

\newenvironment{protocol*}[1]
  {
    \begin{center}
      \hrulefill\\
      \textbf{#1}
  }
  {
    \vspace{-1\baselineskip}
    \hrulefill
    \end{center}
  }

\newcommand{\Tr}{\mbox{\rm Tr}}
\newcommand{\Id}{\ensuremath{\mathop{\rm Id}\nolimits}}
\newcommand{\Es}[1]{\textsc{E}_{#1}}

\newcommand{\C}{\ensuremath{\mathbb{C}}}
\newcommand{\N}{\ensuremath{\mathbb{N}}}
\newcommand{\F}{\ensuremath{\mathbb{F}}}

\newcommand{\setft}[1]{\mathrm{#1}}
\newcommand{\Density}{\setft{D}}
\newcommand{\Pos}{\setft{Pos}}

\newcommand{\Lin}{\setft{L}}

\newcommand{\Ext}{\ensuremath{\mathrm{Ext}}}

\newcommand{\nmExt}{\ensuremath{\mathrm{nmExt}}}
\newcommand{\Adv}{\ensuremath{\mathrm{Adv}}}
\newcommand{\mac}{\ensuremath{\textsc{MAC}}}

\newcommand{\kd}{\ensuremath{\textsc{KeyDerived}}}
\newcommand{\kc}{\ensuremath{\textsc{KeyConfirmed}}}

\newcommand{\mH}{\mathcal{H}}

\newcommand{\Hmin}{H_{\mathrm{min}}}

\newcommand{\guess}{\textsc{guess}}

\newcommand{\logeps}{{\log(1/\eps)}}
\newcommand{\Enc}{{\mathsf{Enc}}}

\nc{\nl}{\nn \\ &=}  
\nc{\nnl}{\nn \\ &}  
\nc{\fot}{\frac{1}{2}} 
\nc{\oo}[1]{\frac{1}{#1}} 
\newcommand{\ben}{\begin{enumerate}}
\newcommand{\een}{\end{enumerate}}
\nc{\mc}{\mathcal}
\nc{\beq}{\begin{equation}}
\nc{\eeq}{\end{equation}}
\nc{\norm}[1]{\L\| #1 \R\|}
\nc{\onenorm}[1]{\L\| #1 \R\|_1} 

\nc{\Ra}{\Rightarrow}
\nc{\zo}{\{0,1\}}	
\makeatother

\usepackage{babel}
  \providecommand{\algorithmname}{Algorithm}
  \providecommand{\definitionname}{Definition}
  \providecommand{\lemmaname}{Lemma}
  \providecommand{\problemname}{Problem}
  \providecommand{\remarkname}{Remark}
\providecommand{\theoremname}{Theorem}

\begin{document}

\title{A Quantum-Proof Non-Malleable Extractor\\[2mm]
\large With Application to Privacy Amplification against Active Quantum Adversaries}
\author{Divesh Aggarwal\thanks{Center of Quantum Technologies, and Department of Computer Science, NUS, Singapore. email: \texttt{dcsdiva@nus.edu.sg}.}  \and Kai-Min Chung\thanks{Institute of Information Science, Academia Sinica, Taipei 11529, Taiwan. email: \texttt{kmchung@iis.sinica.edu.tw}} \and Han-Hsuan Lin\thanks{SPMS, Nanyang Technological University and Centre for Quantum Technologies, Singapore. email: \texttt{linhh@ntu.edu.sg}}\and Thomas
  Vidick\thanks{Department of Computing and Mathematical Sciences,
    California Institute of Technology, Pasadena, USA. email:
    \texttt{vidick@cms.caltech.edu}. Supported by NSF CAREER Grant CCF-1553477, AFOSR YIP award number FA9550-16-1-0495, and the IQIM, an NSF Physics Frontiers Center (NSF Grant PHY-1125565) with support of the Gordon and Betty Moore Foundation (GBMF-12500028).}}
\date{}

\begin{titlepage}
\clearpage
  \maketitle
\thispagestyle{empty}

\begin{abstract}
In privacy amplification, two  mutually trusted parties aim to amplify the secrecy of an initial shared secret $X$ in order to establish a shared private key $K$ by exchanging messages over an insecure communication channel. If the channel is authenticated the task can be solved in a single round of communication using a strong randomness extractor; choosing a quantum-proof extractor allows one to establish security against quantum adversaries. 

In the case that the channel is not authenticated, Dodis and Wichs (STOC'09) showed that the problem can be solved in two rounds of communication using a non-malleable extractor, a stronger pseudo-random construction than a strong extractor. 

We give the first construction of a non-malleable extractor that is secure against quantum adversaries. The extractor is based on a construction by Li (FOCS'12), and is able to extract from source of min-entropy rates larger than $1/2$. Combining this construction with a quantum-proof variant of the reduction of Dodis and Wichs, shown by Cohen and Vidick (unpublished), we obtain the first privacy amplification protocol secure against active quantum adversaries.  
\end{abstract}

\end{titlepage}

\section{Introduction}

\paragraph{Privacy amplification.}
We study the problem of {\em privacy amplification}~\cite{BBR88,Mau92,BBCM95,MW97} (PA). In this problem, two parties, Alice and Bob, share a weak secret $X$ (having min-entropy at least $k$). Using $X$ and an insecure communication channel, Alice and Bob would like to securely agree on a secret key $R$ that is $\eps$-close to uniformly random even to an adversary Eve who may have full control over their communication channel. This elegant problem has multiple applications including biometric authentication, leakage-resilient cryptography, and quantum cryptography.

If the adversary Eve is passive, i.e., she is only able to observe the communication but may not alter the messages exchanged, then there is a direct solution based on the use of a strong seeded randomness extractor $\Ext$~\cite{NZ96}. This can be done by Alice selecting a uniform seed $Y$ for the extractor, and sending the seed to Bob; Alice and Bob both compute the key $R = \Ext(X, Y)$, which is close to being uniformly random and independent of $Y$ by the strong extractor property. The use of a quantum-proof extractor suffices to protect against adversaries holding quantum side information about the secret $X$. 
 
Privacy amplification is substantially more challenging when the adversary is active, i.e. Eve can not only read but also modify messages exchanged across the communication channel. This problem has been studied extensively~\cite{MW97, RW03, DKKRS12, DW09, CKOR10, DLWZ14, CRS12, Li12a, Li12b, Li13, DY13, ADJMR14,Li15, CGL15, Coh15, AHL16, Li17}, yielding constructions that are optimal in any of the parameters involved in the problem, including the min-entropy $k$, the error $\eps$, and the communication complexity of the protocol.

\paragraph{Active adversaries with quantum side information.}
We  consider the problem of active attacks by quantum adversaries. This question arises naturally when privacy amplification is used as a sub-protocol, e.g., as a post-processing step in quantum key distribution (QKD), when it may not be safe to assume that the classical communication channel is authenticated.\footnote{QKD relies on an authenticated channel at other stages of the protocol, and here we only address the privacy amplification part: indeed, PA plays an important role in multiple other cryptographic protocols, and it is a fundamental task that it is useful to address first.} To the best of our knowledge the question was first raised in~\cite{BF11}, whose primary focus is privacy amplification with an additional property of source privacy. Although the authors of~\cite{BF11} initially claimed that their construction is secure against quantum side information, they later realized that there was an issue with their argument, and withdrew their claim of quantum security. The only other work we are aware of approaching the question of privacy amplification in the presence of active quantum adversaries is a paper by one of us~\cite{CV16}. In this paper it is shown that a classical protocol for PA introduced by Dodis and Wichs~\cite{DW09} remains secure against active quantum attacks when the main tool used in the protocol, a non-malleable extractor, is secure against quantum side information (a notion that is also formally introduced in that paper, and to which we return shortly). Unfortunately, the final contribution of~\cite{CV16}, a construction of a quantum-proof non-malleable extractor, also had a flaw in the proof, invalidating the construction. Thus, the problem of quantum-secure active privacy amplification remained open.

It may be useful to discuss the difficulty faced by both these previous works, as it informed our own construction. The issue is related to the modeling of the side information held by the adversary Eve, and how that side information evolves as messages are being exchanged, and possibly modified, throughout the privacy amplification protocol. To explain this, consider the setting for a non-malleable extractor. Here, Alice initially has a secret $X$ (the source), while Eve holds side information $E$, a quantum state, correlated with $X$. Alice selects a uniformly random seed $Y$ and computes $\Ext(X,Y)$. However, in addition to receiving $Y$ (as would already be the case for a strong randomness extractor), Eve is also given the possibility to select an arbitrary $Y'\neq Y$ and receive $\Ext(X,Y')$ as ``advice'' to help her break the extractor --- i.e., distinguish $\Ext(X,Y)$ from uniform. Now, clearly in any practical scenario the adversary may use her side information $E$ in order to guide her choice of $Y'$; thus $Y'$ should be considered as the outcome of a measurement $\{M_y^{y'}\}$, depending on $Y=y$ and performed on $E$, which returns an outcome $Y'=y'$ and a post-measurement state $E'$. This means that the security of the extractor should be considered with respect to the side information $E'$. But due to the measurement, $E'$ may be correlated with both $X$ and $Y$ in a way that cannot be addressed by standard techniques for the analysis of strong extractors. Indeed, even if $E'$ is classical, so that we can condition on its value, $X$ and $Y$ may not be independent after conditioning on $E' = e'$; due to the lack of independence it is unclear whether extraction works. 

The issue seems particularly difficult to accommodate when analyzing extractors based on the technique of ``alternate extraction'', as was attempted in~\cite{BF11,CV16}.
In fact, in the original version of~\cite{BF11} the issue is overlooked, resulting in a flawed security proof. In~\cite{CV16} the authors attempted to deal with the difficulty by using the formalism of quantum Markov chains; unfortunately, there is a gap in the argument and it does not seem like the scenario can be modeled using the Markov chain formalism. Note that in the classical setting the issue does not arise: having fixed $E=e$ we can consider $Y'$ to be a fixed, deterministic function of $Y$ --- there is no $E'$ to consider, and $X$ is independent of both $Y$ and $Y'$ conditioned on $E=e$.
In this paper we do not address the issue, but instead focus on a specific construction of non-malleable extractor whose security can be shown by algebraic techniques sidestepping the difficulty; we explain our approach in more detail below.

\paragraph{Our results.} We show that a non-malleable extractor introduced by Li~\cite{Li12a} in the classical setting is secure against quantum side information. Combining this construction with the protocol of Dodis and Wichs and its proof of security from~\cite{CV16}, we obtain the first protocol for privacy amplification that is secure against active quantum adversaries.  

Before describing our results in more detail we summarize Li's construction and its analysis for the case of classical side information. The construction is based on the inner product function. Let $p$ be a prime, $\F_p$ the finite field with $p$ elements, and $\langle \cdot, \cdot\rangle$ the inner product over $\F_p$. Consider the function $\Ext: \F_p^n \times \F_p^n \to \F_p$ given by $\Ext(X, Y) := \langle X, Y \rangle$, where $X \in \F_p^n$ is a weak secret with min-entropy (conditioned on the adversary's side information) assumed to be greater than $(n \log p)/2$, and $Y$ is a uniformly random and independent seed. For this function to be a non-malleable extractor, it is required that $\Ext(X, Y)$ is close to uniform and independent of $\Ext(X, f(Y))$, where $f$ is any adversarially chosen function such that $f(Y) \neq Y$ for all $Y$. This is clearly not true, since if $f(Y) = cY$ for some $c \in \F_p\setminus \{1\}$, then $\Ext(X, f(Y)) = c \Ext(X, Y)$, and hence we don't get the desired independence. Thus, for such a construction to work, it is necessary to encode the source $Y$ as $\Enc(Y)$, for a well-chosen function $\Enc$, in such a way that $\langle X, \Enc(Y) \rangle - c\cdot \langle X, \Enc(f(Y))\rangle$ is hard to guess. The non-uniform XOR lemma~\cite{CRS12,AHL16} shows that it is sufficient to show that  $\langle X, \Enc(Y) \rangle - c\cdot \langle X, \Enc(f(Y))\rangle = \langle X, \Enc(Y) - c\cdot \Enc(f(Y)) \rangle$ is close to uniform conditioned on  $Y$ and $E$. 
The encoding that we use in this paper (which is almost the same as the encoding chosen by Li) is to take $Y \in \F_{p}^{n/2}$, and encode it as $Y \| Y^2$, which we view as an $n$-character string over $\F_p$, with the symbol $\|$ denoting concatenation of strings and the square taken by first interpreting $Y$ as an element of $\F_{p^{n/2}}$. Then it is not difficult to show that for any function $f$ such that $f(Y) \neq Y$ and any $c$, we have that $(Y\| Y^2) - (c \cdot f(Y) \| c\cdot f(Y)^2)$ (taking the addition coordinatewise) has min-entropy almost $(n\log p)/2$. Thus, provided $X$ has sufficiently high min-entropy and using the fact that $X$ and $(Y\| Y^2) - (c \cdot f(Y) \| c\cdot f(Y)^2)$ are independent conditioned on $E$, the strong extractor property of the inner product function gives the desired result.\footnote{This description is a little different from Li's description since he was working with a field of size $2^n$, but we find it more convenient to work with a prime field.}



Our main technical result is a proof of security of Li's extractor, against quantum side information. We show the following (we refer to Definition~\ref{def:nmext} for the formal definition of a quantum-proof non-malleable extractor):


\begin{theorem}\label{thm:main}
Let $p\neq 2$ be a prime. Let $n$ be an even integer. Then for any $\eps>0$ the function $\nmExt(X,Y):\F_p^n \times \F_p^{n/2} \ra \F_p$ given by $\langle X, Y\|Y^2 \rangle$ is an $(\L(\frac{n}{2}+6\R){\log p} -1+4\log\oo{\eps},\eps)$  quantum-proof non-malleable extractor. 
\end{theorem}

We give the main ideas behind our proof of security for this construction, highlighting the points of departure from the classical analysis. Subsequently, we explain the application to privacy amplification. 

\paragraph{Proof ideas.} 
We begin by generalizing the first step of Li's argument, the reduction provided by the non-uniform XOR lemma, to the quantum case. An XOR lemma with quantum side information is already shown in~\cite{KK12}, where the lemma is used to show security of the  inner product unction as a two-source extractor against quantum side information. This version is not sufficient for our purposes, and we establish the following generalization, which may be of independent interest (we refer to Section~\ref{sec:xor} for additional background and definitions):

\begin{lemma}\label{lem:xor2}
Let $p$ be a prime power and $t$ an integer. Let $\rho_{X_0XE}$ be a ccq state with $X_0 \in \F_p$ and $X=(X_1,\dots,X_t)\in \F^t_p$. For all $a=(a_1,\dots,a_t)\in \F^t_p$, define a random variable $Z=X_0+\vev{a,X}=X_0+\sum_{i=1}^t a_i X_i$. Let $\eps\geq 0$ be such that for all $a$,  $\fot \onenorm{\rho^a_{ZE}-U_Z\otimes \rho_E}\leq \epsilon$.  Then
\ba
 \fot\big\|\rho_{X_0XE}-U_{X_0}\otimes \rho_{XE}\big\|_1 &\leq p^{\frac{t+1}{2}} \sqrt{\frac{\epsilon}{2}}\;.
\ea
\end{lemma}

XOR lemmas are typically proved via Fourier-based techniques (including the one in~\cite{KK12}). Here we instead rely on a collision probability-based argument inspired from~\cite{AHL16}. We prove  Lemma~\ref{lem:xor2} by observing that such arguments generalize to the quantum setting, as in the proof of the quantum leftover hash lemma in~\cite{leftoverhash}. 


Based on the XOR lemma (used with $t=1$), following Li's arguments it remains to show that the random variable $\vev{X,g(Y,Y')}\in\F_p$, where $g(Y,Y') = Y\|Y^2 - c(Y'\| Y'^2) \in \F_p^n$, is close to uniformly distributed from the adversary's point of view, specified by side information $E'$, for every $c \neq 0 \in \F_p$. 
As already mentioned earlier, this cannot be shown by  a reduction to the security proof of the inner product function as a two-source extractor against side information, as $X$ and $g(Y,Y')$ are \emph{not} independent (not even conditioned on the value of $E'$ when $E'$ is classical).

Instead, we are led to a more direct analysis which proceeds by formulating the problem as a communication task.\footnote{The correspondence between security of quantum-proof strong extractors and communication problems has been used repeatedly before, see e.g.~\cite{GKKRW07,KK12}.  } 
We relate the task of breaking our construction --- distinguishing $\vev{X,g(Y,Y')}$ from uniform --- to success in the following task. Alice is given access to a random variable $X$, and Bob is given a uniformly random $Y$. Alice is allowed to send a quantum message $E$, correlated to $X$, to Bob. Bob then selects a $Y' \neq Y$ and returns a value $b\in \F_p$. The players win if $b=\vev{X,g(Y,Y')}$. 
Based on our previous reductions it suffices to show that no strategy can succeed with probability substantially higher than random in this game, unless Alice's initial message to Bob contains a large amount of information about $X$; more precisely, unless the min-entropy of $X$, conditioned on $E$, is less than half the length of $X$. 

Note that the problem as we formulated it does not fall in standard frameworks for communication complexity. In particular, it is a relation problem, as Bob is allowed to choose the value $Y'$ to which his prediction $b$ applies. This seems to prevent us from using any prior results on the communication complexity of the inner product function, and we develop an ad-hoc proof which may be of independent interest. We approach the problem using the ``reconstruction paradigm'' (used in e.g.~\cite{DVPR11}), which amounts to showing that from any successful strategy of the players one may construct a measurement for Bob which completely ``reconstructs'' $X$, given $E$; if this can be achieved with high enough probability it will contradict the min-entropy assumption on $X$, via its dual formulation as a guessing probability~\cite{minentropyguess}. We show this by running Bob's strategy ``in superposition'', and applying a Fourier transform to recover a guess for $X$. This argument is similar to one introduced in~\cite{CDNT97,NS}. We refer to Section~\ref{sec:com-game} for more detail. 


\paragraph{Application to privacy amplification.}
Finally we discuss the application of our quantum-proof non-malleable extractor to the problem of privacy amplification against active quantum attacks, which is our original motivation. The application is based on a breakthrough result by Dodis and Wichs~\cite{DW09}, who were first to show the existence of a two-round PA protocol with optimal (up to constant factors) entropy loss $L = \Theta(\logeps)$, for any initial min-entropy $k$. This was achieved by defining and showing the existence of non-malleable extractors with very good parameters.

The protocol from~\cite{DW09} is recalled in Section~\ref{sec:pa}. The protocol proceeds as follows. Alice sends a uniformly random seed $Y$ to Bob over the communication channel, which is controlled by Eve. Bob receives a possibly modified seed $Y'$. Then Alice computes a key $K = \nmExt(X, Y)$, and Bob computes $K' = \nmExt(X, Y')$. In the second round, Bob generates another uniformly random seed $W'$, and sends $W'$ together with $T' = \mac_{K'}(W')$ to Alice, where $\mac$ is a one-time message authentication code. Alice receives a possibly modified $T, W$ and checks whether $T = \mac_K(W)$. If yes, then the shared secret between Alice and Bob is $\Ext(X, W) = \Ext(X, W')$ with overwhelming probability, where $\Ext$ is any strong seeded extractor. 

The security of this protocol intuitively follows from the following simple observation. If the adversary does not modify $Y$, then $K' = K$, and so $W'$ must be equal to $W$ by the security of the $\mac$. If $Y' \neq Y$, then by the non-malleability property of $\nmExt$, $K$ is uniform and independent of $K'$, and so it is impossible for the adversary to predict $\mac_K(W)$ for any $W$ even given $K'$ and $W'$. 

Since~\cite{DW09} could not construct an explicit non-malleable extractor, they instead defined and constructed a so called a look-ahead extractor, which can be seen as a weakening of the non-malleability requirement of a non-malleable extractor. This was done by using the alternating extraction protocol by Dziembowski and Pietrzak~\cite{DP07}. 

In~\cite{CV16}, Dodis and Wichs' reduction is extended to the case of quantum side information, provided that the non-malleable Extractor $\nmExt$ used in the protocol satisfies the approriate definition of quantum non-malleability, and $\Ext$ is a strong quantum-proof extractor. Based on our construction of a quantum-proof non-malleable extractor (Theorem~\ref{thm:main}) we immediately obtain a PA protocol that is secure as long as the initial secret $X$ has a min-entropy rate of (slightly more than) half. The result is formalized as Corollary~\ref{cor:dw-pa} in Section~\ref{sec:pa}. 

%

\paragraph{Future work.} 
There have been a series of works in the classical setting~\cite{DLWZ14,CRS12,Li12a, Li12b, DY13, Li15, CGL15, Coh15, AHL16, Li17} that have given privacy amplification protocols (via constructing non-malleable extractors or otherwise) that achieve near-optimal parameters. In particular, Li~\cite{Li17} constructed a non-malleable extractor that works for min-entropy $k = \Omega(\log n +  \logeps \log \logeps)$, where $\eps$ is the error probability. 

Our quantum-proof non-malleable extractor requires the min-entropy rate of the initial weak secret to be larger than $1/2$. We leave it as an open question whether one of the above-mentioned protocols that work for min-entropy rate smaller than $1/2$ in the classical setting can be shown secure against quantum side information. 

\section{Preliminaries}\label{sec:prelim}

\subsection{Notation}
\label{sec:notation}

For $p$ a prime power we let $\F_p$ denote the finite field with $p$ elements. For any positive integer $n$, there is a natural bijection $\phi: \F_p^n \mapsto \F_{p^n}$ that preserves group addition and scalar multiplication, i.e., the following hold:
\begin{itemize}
\item For all $c \in \F_p$, and for all $x \in \F_p^n$, $\phi(c \cdot x) = c \cdot \phi(x)$. 
\item For all $x_1, x_2 \in \F_p^n$, $\phi(x_1) + \phi(x_2) = \phi(x_1 + x_2)$.
\end{itemize}
We use this bijection to define the square of an element in $\F_p^n$, e.g. for $y \in \F_p^n$ 
\ba
 y^2=\phi^{-1}\L((\phi(y))^2\R)\;.
\ea 
We write $\langle \cdot,\cdot \rangle$ for the inner product over $\F_p^n$. $\log$ denotes the logarithm with base $2$.

We write $\mH$ for an arbitrary finite-dimensional Hilbert space, $\Lin(\mH)$ for the linear operators on $\mH$, $\Pos(\mH)$ for positive semidefinite operators, and $\Density(\mH)\subset \Pos(\mH)$ for positive semidefinite operators of trace $1$ (\emph{density matrices}). A linear map $\Lin(\mH)\to\Lin(\mH')$ is CPTP if it is completely positive, i.e. $T\otimes \Id(A)\geq 0$ for any $d\geq 0$ and $A\in\Pos(\mH\otimes\C^d)$, and trace-preserving.

We use capital letters $A,B,E,X,Y,Z,\ldots$ to denote quantum or classical random variables. Generally, the letters near the beginning of the alphabet, such as $A,B,E$, represent quantum variables (density matrices on a finite-dimensional Hilbert space), while the letters near the end, such as $X,Y,Z$ represent classical variables (ranging over a finite alphabet). We sometimes represent classical random variables as density matrices diagonal in the computational basis, and write e.g. $(A ,B  , \dots, E )_\rho$ for the density matrix $\rho_{A ,B, \dots, E}$. For a quantum random variable $A$, we denote $\mH_{A}$ the Hilbert space on which the associated density matrix $\rho_A$ is supported, and $d_A$ its dimension. If $X$ is classical we loosely identify its range $\{0,\ldots,d_X-1\}$ with the space $\mH_X$ spanned by $\{\ket{0}_X,\ldots,\ket{d_X-1}_X\}$. We denote $I_A$ the identity operator on $\mH_A$.   When a density matrix specifies the states of two random variables, one of which is classical and the other is quantum, we call it a classical-quantum(cq)-state. A cq state
 $(X,E)_\rho$ takes the form
\begin{align}
&\rho_{XE}=\sum_x \proj{x}_X\otimes \rho^x_E\;,   \nn
\end{align}
 where the summation is over all $x$ in the range of $X$ and $\{\rho^x_E\}$ are positive semidefinite matrices with $\tr \rho_E^x=p_x$, where $p_x$ is the probability of getting the outcome $x$ when measuring the $X$ register. Similarly, a ccq state $(X,Y,E)_\sigma$ is a density matrix over two classical variables and one quantum variable, e.g. $\sigma_{XYE}=\sum_{x,y} \proj{x}_X\otimes \proj{y}_Y\otimes \sigma^{xy}_E$. We will sometimes add or remove random variables from an already-specified density matrix. When we omit a random variable, we mean the reduced density matrix, e.g. $(Y,E)_\sigma= \tr_X( \sigma_{XYE})$. When we introduce a classical variable, we meant that the  classical variable is computed into another classical register. For example, for a function $F(\cdot,\cdot)$ on variables $X,Y$,
 \ba
 (F(X,Y),X,Y,E)_\sigma = \sum_{f,x,y} {\delta}({f,F(x,y)}) \proj{f}\otimes\proj{x}\otimes\proj{y}\otimes \sigma^{xy}_E\;, \nn
 \ea
 where ${\delta}(\cdot,\cdot)$ is the Kronecker delta function, and the summation over $f$ is taken over the range of $F$.  When $F$ is a random function, the density matrix is averaged over the appropriate probability distribution.

We use $U_\Sigma$ to denote the uniform distribution over a set $\Sigma$. For $m$-bit string $\zo^m$, we abbreviate $U_{\zo^m}$ as $U_m$. For a classical random variable $X$, $U_X$ denote the uniform distribution over the range of $X$. 

For $p\geq 1$ we write $\norm{\cdot}_p$ for the Schatten $p$-norm (for a normal matrix, this is the $p$-norm of the vector of singular values). We write $\norm{\cdot}$ for the operator norm.  

We write $\approx_\eps$ to denote that two density matrices are $\eps$-close to each other in trace distance. For example,  $(X,E)_\rho \approx_\eps (U_X,E)_\rho$ means $\fot \norm{\rho_{XE}- U_X \otimes \rho_{E}}_1 \leq \eps$.
Note that in case both $X$ and $E$ are classical random variables, this reduces to the statistical distance. 

\subsection{Quantum information}

The min-entropy of a classical random variable $X$ conditioned on quantum side information $E$ is defined as follows.

\begin{definition}[Min-entropy]\label{def:min-entropy}
Let $\rho_{XE} \in \Density(\mH_X \otimes \mH_E)$ be a cq state. The \emph{min-entropy} of $X$ conditioned on $E$ is defined as
  \begin{equation*}
    \Hmin({X|E})_\rho = \max \{\lambda \geq 0 :  \exists \sigma_E \in \Pos(\mH_E), \tr\L(\sigma_E \R)\leq1, \, \mathrm{s.t.}\,\, 2^{-\lambda} I_X \otimes \sigma_E \geq \rho_{XE}\}.
  \end{equation*}
When the state $\rho$ with respect to which the entropy is measured is clear from context we simply write $\Hmin({X|E})$ for $\Hmin({X|E})_\rho$.	
\end{definition}

\begin{definition}[$(n,k)$-source]
A cq state $\rho_{XE}$ is an $(n,k)$-source if $n= \log d_X$ and $\Hmin({X|E}))_\rho \geq k$.
\end{definition}

It is often convenient to consider the \emph{smooth} min-entropy, which is obtained by maximizing the min-entropy over all cq states in an $\eps$-neighborhood of $\rho_{XE}$. The definition of this neighborhood depends on a choice of metric; the canonical choice is the ``purified distance''. Since this choice will not matter for us we defer to~\cite{tomamichel2015quantum} for a precise definition.

\begin{definition}\label{prelim:def:smooth-min-entropy}
  Let $\eps \geq 0$ and $\rho_{XE} \in \Density(\mH_X \otimes\mH_E)$ a cq state. The
  \emph{$\eps$-smooth min-entropy} of $X$ conditioned on $E$ is defined as
  \begin{equation*}
    \Hmin^\eps(X|E)_\rho = \max_{\sigma_{XE} \in \mathcal{B}(
      \rho_{XE},\eps) } \Hmin(X|E)_\sigma,
  \end{equation*}
	where $\mathcal{B}(
      \rho_{XE},\eps) $ is the ball of radius $\eps$ around $\rho_{XE}$, taken with respect to the purified distance.
\end{definition}

Rather than using Definition~\ref{def:min-entropy}, we will most often rely on an operational expression for the min-entropy stated in the following lemma from~\cite{minentropyguess}.

\begin{lemma}[Min-entropy and guessing probability]\label{lem:minentropyguess}
For a cq state $\rho_{XE}\in \Density(\mH_X\otimes\mH_E)$, the guessing probability is defined as the probability to correctly guess $X$ with the optimal strategy to measure $E$, i.e.
\ba
p_{guess}(X|E)_\rho= \sup_{\{M_x\}}\sum_x p_x \tr \L( M_x \rho^x_E\R)\;,
\ea  
where ${\{M_x\}}$ is a positive operator-valued measure (POVM) on $\mH_E$. Then the guessing probability is related to the min-entropy by 
\ba
p_{guess}(X|E)_\rho= 2^{-\Hmin(X|E)_\rho}\;.
\ea
\end{lemma}

\subsection{Extractors}

We first give the definition of a strong quantum-proof extractor. Recall the notation $(X,E)_\rho \approx_\eps (X',E')_\rho$ for $\fot \|\rho_{XE}-\rho_{X'E'}\|_1 \leq \eps$, and $U_m$ for a random variable uniformly distributed over $m$-bit strings.

\begin{definition}\label{def:ext}
 Let $k$ be an integer and $\eps\geq 0$.
  A function $\Ext: \mH_X \times \mH_Y \to \mH_Z$ is a
 strong $(k,\eps)$ quantum-proof extractor if for all cq
  states $\rho_{XE}\in\Density(\mH_X\otimes \mH_E)$ with $\Hmin(X|E) \geq k$,
  and for a classical uniform $Y\in \mH_Y$ independent of $\rho_{XE}$,
$$  (\Ext(X,Y),Y,E)_\rho \approx_\eps (U_Z,Y,E)_\rho\;.$$
 \end{definition}

There are known explicit constructions of strong quantum-proof extractors. 

\begin{theorem}[\cite{leftoverhash}]\label{thm:quantum-LHL} 
For any integers $d_X,k$ and for any $\eps > 0$ there exists an explicit strong $(k,\eps)$ quantum-proof extractor $\Ext \colon \{0,\ldots,d_X-1\} \times \{0,\ldots,d_Y-1\} \to \{0,\ldots,d_Z-1\}$ with $\log d_Y = O(\log d_X)$ and $\log d_Z = k - O(\log (1/\eps) - O(1)$.
\end{theorem}

%

We use the same definition of non-malleable extractor against quantum side information that was introduced in the work~\cite{CV16}. The definition is a direct generalization of the classical notion of non-malleable extractor introduced in~\cite{DW09}. The first step is to extend the notion that the adversary may query the extractor on any \emph{different} seed $Y'$ than the seed $Y$ actually used to the case where $Y'$ may be generated from $Y$ as well as quantum side information held by the adversary.

%

\begin{definition}[Map with no fixed points]\label{def:nfp}
Let $\mH_Y$, $\mH_E$ and $\mH_{E'}$ denote finite-dimensional Hilbert spaces.
We say that a CPTP map $T:\Lin(\mH_Y\otimes \mH_E)\to \Lin(\mH_Y\otimes \mH_{E'})$ has \emph{no fixed points} if for all $\rho_E\in \Density(\mH_E)$ and all computational basis states $\ket{y}\in \mH_Y$ it holds that
$$ \bra{y}_Y \,\Tr_{\mH_{E'}}\big( T\big(\proj{y}_Y\otimes \rho_E\big)\big)\,\ket{y}_Y\,  =\, 0\;.$$
\end{definition}

The following definition is given in~\cite{CV16}:

\begin{definition}[Non-mallleable extractor]\label{def:nmext}
Let $\mH_X$, $\mH_Y$, $\mH_Z$ be finite-dimensional Hilbert spaces, of respective dimension $d_X$, $d_Y$, and $d_Z$. Let $k\leq \log d_X$ and $\eps>0$. A function 
$$\nmExt \,:\, \{0,\ldots,d_X-1\}\times\{0,\ldots,d_Y-1\} \to \{0,\ldots,d_Z-1\}$$
is a $(k,\eps)$ quantum-proof non-malleable extractor if for every cq-state $(X,E)_\rho$ on  $\mH_X \otimes \mH_E$ such that $\Hmin(X|E)_\rho \geq k$ and any CPTP map $\Adv: \Lin(\mH_Y\otimes \mH_E) \to \Lin(\mH_Y \otimes \mH_{E'})$ with no fixed points,
$$ \big\| \sigma_{\nmExt(X,Y) \nmExt(X,Y') YY' E'} - U_Z \otimes \sigma_{\nmExt(X,Y') YY'E'}\big)\big\|_1 \leq \eps\;,$$
where
\begin{equation}\label{eq:adv-state}
 \sigma_{YY'XE'} \,=\, \frac{1}{d_Y}\sum_{y} \proj{y}_Y \otimes (I_X \otimes \Adv)( \proj{y}_Y \otimes \rho_{XE} )
\end{equation}
and $\sigma_{\nmExt(X,Y) \nmExt(X,Y') YY' E'}$ is obtained from $\sigma_{YY'XE'}$ by (classically) computing $\nmExt(X,Y)$ and $\nmExt(X,Y')$ in ancilla registers and tracing out $X$.
\end{definition}

\subsection{H\"{o}lder's inequality  } 
We use the following H\"{o}lder's inequality for matrices. For a proof, see e.g. \cite{Bhatia}.
\begin{lemma}[H\"{o}lder's inequality]\label{lem:holder}
For any $n\times n$ matrices $A$, $B$, $C$ with complex entries, and real numbers $r,s,t > 0$ satisfying $\oo{r}+\oo{s}+\oo{t}=1$,
\ba
\onenorm{ABC}  \leq \onenorm{|A|^r}^{1/r} \onenorm{|B|^s}^{1/s} \onenorm{|C|^t}^{1/t} \;.
\ea
\end{lemma}

\section{Quantum XOR lemma} \label{sec:xor}

In this section we prove two XOR lemmas with quantum side information. First we prove a more standard version, Lemma~\ref{lem:xor}, in Section~\ref{subsection:xor}. We prove a  non-uniform version, Lemma~\ref{lem:xor2}, in Section~\ref{subsection:xor2}. Lemma~\ref{lem:xor} is not actually used in the analysis of our non-malleable extractor, but we include it as a warm-up to our non-uniform XOR lemma. Since XOR lemmas often play a fundamental role, both lemma might be of independent interest. The proofs are based on quantum collision probability techniques\footnote{The term ``quantum collision probability'' is ours.} from~\cite{leftoverhash} to transform a classical collision probability-based proof into one that also allows for quantum side information.
 When restricted to $\F_2$, our standard XOR lemma, Lemma~\ref{lem:xor}, is very similar to Lemma~10 of \cite{KK12}, although the result from~\cite{KK12} provides a tighter bound in this case.\footnote{\cite{KK12} provides a bound of $p^{2t}\eps^2$ (for $p=2$), while ours scales as $p^t \eps$, a quadratic loss.} However our result applies to $\F_p$, while it is unclear whether the proof of \cite{KK12} generalizes to $p>2$. Our result is based on the use of the collision probability, while \cite{KK12} uses Fourier analysis. The idea of non-uniform XOR lemma is natural in the context of non-malleable extractors, and has been explored in \cite{Li12a,CRS12,AHL16}.  Our non-uniform XOR lemma generalizes a restricted version of Lemma 3.15 of \cite{Li12a} to $\F_p$ with quantum side information.\footnote{Compared to~\cite[Lemma 3.15]{Li12a} we have $m=1$ and $n=t$.}

The quantum collision probability is defined as follows. 

\begin{definition} [Quantum collision probability]
Let $\rho_{AB} \in \Density(\mH_A\otimes \mH_B)$ and $\sigma_B \in \Density(\mH_B)$. The collision probability of $\rho_{AB}$, conditioned on $\sigma_B$, is defined as
\ba
\Gamma_c(\rho_{AB}|\sigma_B) \equiv \tr\L(\rho_{AB} (I_A\otimes \sigma_B^{-1/2}) \R)^2
\;,\ea
where $\sigma_B \in \Density(\mH_B)$.
\end{definition}
A careful reader might notice that $\Gamma_c \leq 1$ is not generally true, so calling $\Gamma_c$ collision \emph{probability} seems misleading. We give a general definition which allows arbitrary states $\rho_{AB}$ and $\sigma_B$ to match the existing literature, but here we always consider cq states $\rho_{AB}$ and take $\sigma_B=\rho_B$. We prove in Corollary~\ref{cor:gamma-ub} that $\Gamma_c \leq 1$ in such cases. $\Gamma_c(\rho_{AB}|\sigma_B)$ also reduces to the classical collision probability when both of $A,B$ are classical and $\sigma_B=\rho_B$.

{We will often use the following relation, also taken from~\cite{leftoverhash}, valid for any $\rho_{AB} \in \Density(\mH_A\otimes \mH_B)$: 
\begin{equation}\label{eq:trace-to-collision}
\tr\L((\rho_{AB}-U_A \otimes \rho_B )(I_A\otimes \rho_B^{-1/2})\R)^2  \,=\,\Gamma_c\big(\rho_{AB}|\rho_B\big) - \oo{d_A}\;,
\end{equation}
which can be verified by expanding the square: 
\begin{align*}
&\tr\Big((\rho_{AB}-U_A \otimes \rho_B )(I_A\otimes \rho_B^{-1/2})\Big)^2  \\
&= \tr\L(\rho_{AB} \, \rho_B^{-1/2}\R)^2 - 2 \tr \L(\rho_{AB} \, \rho_B^{-1/2} (U_A \rho_B) \rho_B^{-1/2}\R) +\tr\L((U_A \rho_B)  \rho_B^{-1/2}\R)^2 \\
&=\Gamma_c(\rho_{AB}|\rho_B) - \oo{d_A}\;.
\end{align*}} 

\subsection{Standard quantum XOR lemma}\label{subsection:xor}

In this section, we prove our quantum generalization of the standard XOR lemma, Lemma~\ref{lem:xor}, which roughly states that given a random variable $X\in \F_p^t$, if $\vev{a,X}$ is close to uniform for all $a\neq 0$, then $X$ is close to uniform.

\begin{lemma}[XOR lemma]\label{lem:xor}
Let $p$ be a prime power, $t$ an integer and $\eps\geq 0$. Let $X=(X_1,\dots,X_t)\in \F^t_p$ be a classical random variable, and $E$ some correlated quantum variable. For all $a=(a_1,\dots,a_t)\in \F^t_p$, define a random variable $Z=\vev{a,X}=\sum_{i=1}^t a_i X_i$. If for all $a \neq 0$,  $\fot \onenorm{\rho^a_{ZE}-U_Z\otimes \rho_E}\leq \epsilon$, then
\ba
\fot \big\|\rho_{XE}-U_X\otimes \rho_E\big\|_1\, \leq \,\frac{p^{t/2}}{\sqrt 2}\sqrt \epsilon\;.
\ea
\end{lemma}

To prove Lemma~\ref{lem:xor}, we  bound the collision probability as a function of the trace distance in Lemma~\ref{lem:2qb}. To prove Lemma~\ref{lem:2qb} we use the following lemma for ``conditioned'' cq states.
\begin{lemma}\label{lem:eig} For a cq state $\rho_{XE}$ with reduced density matrix $\rho_E$,
\ba
0\leq \L (I_X\otimes \rho_E^{-1/2}\R) \rho_{XE} \L (I_X\otimes \rho_E^{-1/2}\R) \leq  I_{XE}\;.
\ea
\end{lemma}
\begin{proof}
 We proceed by bounding the eigenvalues of the middle term. Using the fact that $\rho_{XE}$ is a cq state, we observe that 
\ba
\rho_E^{-1/2} \rho_{XE} \rho_E^{-1/2}=\sum_x \proj{x} \rho_E^{-1/2} \rho^x_{E} \rho_E^{-1/2}
\ea
is block diagonal. Combining with the fact that for all states $ \;\ket{\phi} \in \mH_E $ and $x$ in the range of $X$, 
\ba
0 \leq \vev{\phi|\rho_E^{-1/2} \rho^x_{E} \rho_E^{-1/2}|\phi} \leq  \vev{\phi|\rho_E^{-1/2} \rho_{E} \rho_E^{-1/2}|\phi} \leq 1\;.
\ea
We conclude that the eigenvalues of $\rho_E^{-1/2} \rho_{XE} \rho_E^{-1/2}$ are between 0 and 1.
\end{proof}
As a corollary of Lemma~\ref{lem:eig}, we can upper bound the quantum collision probability by $1$ in certain settings.
\begin{corollary}[Upper bound of quantum collision probability]\label{cor:gamma-ub} For any cq state $\rho_{XE}$,
\ba
\Gamma_c(\rho_{XE}|\rho_E) \leq 1 \, .
\ea
\end{corollary}
\begin{proof}
Using the definition, 
\ba
\Gamma_c(\rho_{XE}|\rho_E) &= \tr\L( \rho_{XE} \L (I_X\otimes \rho_E^{-1/2}\R) \rho_{XE} \L (I_X\otimes \rho_E^{-1/2}\R) \R) \nn \\
& \leq \tr\L(\rho_{XE} I_{XE} \R) \nn\\
&\leq 1\;,
\ea
where the first inequality follows from Lemma~\ref{lem:eig}.
\end{proof}

The following lemma bounds the  collision probability as a function of the trace distance.

\begin{lemma}[Bounding collision probability with trace distance]\label{lem:2qb}
Let $ \rho_{XE}$ be a cq state. If 
\ba
\fot \onenorm{\rho_{XE}-U_X \otimes \rho_E}=\epsilon\;, 
\ea 
then 
\ba
\frac{4  \epsilon^2 }{d_X} \leq \Gamma_c(\rho_{XE}|\rho_E) -\oo{d_X} \leq 2 \epsilon\L(1- \oo{d_X}\R)\;.
\ea
\end{lemma}
\begin{proof}
The first inequality follows from Lemma  4 of  \cite{leftoverhash}, we repeat it here for completeness. We use H\"{o}lder's inequality (Lemma~\ref{lem:holder}) with $r=t=4, \, s=2, \, A=C=I_X \otimes \rho_E^{1/4}$, and 
$$B=\L(I_X \otimes \rho_E^{-1/4}\R)   \L(\rho_{XE}-U_X \rho_E\R)   \L(I_X \otimes \rho_E^{-1/4}\R)\;.$$
 This leads to
\ba
 2 \epsilon &=  \onenorm{\rho_{XE}-U_X \rho_E} \nl
\onenorm{ABC} \nn \\
 & \leq \onenorm{A^4}^{1/4} \onenorm{B^2}^{1/2} \onenorm{C^4}^{1/4} \nl
\sqrt{d_X \tr \L( \L(\rho_{XE}-U_X \rho_E\R)  \L(I_X \otimes \rho_E^{-1/2}\R) \R)^2} \nl
\sqrt{d_X\L( \Gamma_c(\rho_{XE}|\rho_E) -\oo{d_X} \R)} \, ,
\ea
 where we used Eq.~\eqref{eq:trace-to-collision} in the last line. Squaring both sides and dividing by $d_A$, we get the desired inequality.
 
For the second inequality, we use  Lemma~\ref{lem:eig} to show the following sequence of inequalities
\ba \label{eq:bound-abs-delta}
&0\leq \L (I_X\otimes \rho_E^{-1/2}\R) \rho_{XE} \L (I_X\otimes \rho_E^{-1/2}\R) \leq  I_{XE} \nn \\
&\Ra  -\oo{d_X} I_{XE}\leq \L (I_X\otimes \rho_E^{-1/2}\R) \rho_{XE} \L (I_X\otimes \rho_E^{-1/2}\R) -U_X\otimes I_E\leq \L(1-\oo{d_X}\R) I_{XE} \nn \\
&\Ra \L| \rho_E^{-1/2}\rho_{XE} \rho_E^{-1/2}-U_X I_E \R| \leq \L(1-\oo{d_X}\R) I_{XE}\;.  
\ea 
Starting with  Eq.~\eqref{eq:trace-to-collision}, we get
\ba
 \quad\; \Gamma_c(\rho_{XE}|\rho_E) -\oo{d_X}  &=
{ \tr \L( \L(\rho_{XE}-U_X \rho_E\R)  \L(I_X \otimes \rho_E^{-1/2}\R) \R)^2} \nl
{ \tr \L( \L(\rho_{XE}-U_X \rho_E\R)  \L(\rho_E^{-1/2} \rho_{XE} \rho_E^{-1/2} - U_X I_E \R) \R)} \nn \\
 &\leq { \tr \L( \L|\rho_{XE}-U_X \rho_E\R|  \L|\rho_E^{-1/2} \rho_{XE} \rho_E^{-1/2} - U_X I_E \R|\R)} \nn \\
 &\leq { \tr \L( \L|\rho_{XE}-U_X \rho_E\R| \L(1- \oo{d_X}\R) I_{XE} \R)} \nl
 2 \epsilon  \L(1- \oo{d_X}\R)\;,
\ea
where we used Eq.~\eqref{eq:bound-abs-delta} to go from the third line to the fourth line. 
\end{proof}

Now we are ready to prove Lemma~\ref{lem:xor}. The proof idea is to start from the trace distance of $X$ to uniform, apply Lemma~\ref{lem:2qb} to get an upper bound in terms of the collision probability of $X$, apply Eq.~\eqref{eq:trace-to-collision} and expand the square to express the collision probability of $X$ in terms of the collision probability of $\vev{a,X}$, and finally apply Lemma~\ref{lem:2qb} again to get an upper bound in terms of the trace distance of $\vev{a,X}$ to uniform.

\newtheorem*{lem:xor}{Lemma \ref{lem:xor}}
\begin{lem:xor}[restated]
Let $p$ be a prime power, $t$ an integer and $\eps\geq 0$. Let $X=(X_1,\dots,X_t)\in \F^t_p$ be a classical random variable, and $E$ some correlated quantum variable. For all $a=(a_1,\dots,a_t)\in \F^t_p$, define a random variable $Z=\vev{a,X}=\sum_{i=1}^t a_i X_i$. If for all $a \neq 0$,  $\fot \onenorm{\rho^a_{ZE}-U_Z\otimes \rho_E}\leq \epsilon$, then
\ba
\fot \big\|\rho_{XE}-U_X\otimes \rho_E\big\|_1\, \leq \,\frac{p^{t/2}}{\sqrt 2}\sqrt \epsilon\;.
\ea
\end{lem:xor}
 
\begin{proof}
Note that
\ba
\rho^a_{ZE}=\sum_z \proj{z} \otimes \sum_x {\delta}(z-\vev{a,x} , 0)\rho^x_E\;.
\ea
We start by relating the collision probability of $Z$ and $X$. Using Eq.~\eqref{eq:trace-to-collision},
%
\ba
  \Gamma_c(\rho^a_{ZE}|\rho_E)-\oo{p}&=\tr\L[(\rho^a_{ZE}-U_Z\rho_E)I_Z\otimes \rho_E^{-1/2} \R]^2 \nl
\tr \L[ \sum_z\proj{z} \sum_x\L({\delta}\L(z-\vev{a,x},0\R)-\oo{p} \R)  \rho^x_E \L(I_Z \otimes \rho_E^{-1/2}\R)  \R]^2  \nl
\sum_z \tr \L[\sum_x \L( {\delta}\L(z-\vev{a,x},0\R)-\oo{p} \R)  \rho^x_E \rho_E^{-1/2}  \R]^2 \nl
\sum_{z,x,x'} \L[  {\delta}\L(z-\vev{a,x},0\R)  {\delta}\L(z-\vev{a,x'},0\R)-\frac{2}{p} {\delta}\L(z-\vev{a,x},0\R)+\frac{1}{p^2} \R]\tr\L(\rho^x_E\rho_E^{-1/2}\rho^{x'}_E\rho_E^{-1/2} \R) \nl
\sum_{x,x'} \L[  {\delta}\L(\vev{a,x-x'},0\R)  -\frac{1}{p} \R]\tr\L(\rho^x_E\rho_E^{-1/2}\rho^{x'}_E\rho_E^{-1/2} \R) \nl
\sum_{x} \L(  1 -\frac{1}{p} \R)\tr\L(\rho^x_E\rho_E^{-1/2}\rho^x_E\rho_E^{-1/2} \R) +
\sum_{x\neq	x'} \L[  {\delta}\L(\vev{a,x-x'},0\R)  -\frac{1}{p} \R]\tr\L(\rho^x_E\rho_E^{-1/2}\rho^{x'}_E\rho_E^{-1/2} \R)\;.
\ea
Averaging over $a\in \F^t_p$, we get:
\ba
\Es{a}  \L(\Gamma_c(\rho^a_{ZE}|\rho_E)-\oo{p}\R) &=   \L(  1 -\frac{1}{p}\R) \Gamma_c(\rho_{XE}|\rho_E) + \sum_{x\neq	x'} \Es{a} \L[  {\delta}\L(\vev{a,x-x'},0\R)  -\frac{1}{p} \R]\tr\L(\rho^x_E\rho_E^{-1/2}\rho^x_E\rho_E^{-1/2} \R) \nl
 \L(  1 -\frac{1}{p}\R) \Gamma_c(\rho_{XE}|\rho_E)\;. 
\ea
Dividing both sides by $\L(  1 -\frac{1}{p} \R)$ gives
\ba \label{eq:xorcollision}
\Gamma_c(\rho_{XE}|\rho_E)=\frac{p}{p-1}\Es{a}  \L(\Gamma_c(\rho^a_{ZE}|\rho_E)-\oo{p}\R)\;.
\ea
%

Having the most computational-heavy step done, we put the pieces together to prove our XOR lemma. 
\ba
\big\|\rho_{XE}-U_X\otimes \rho_E\big\|_1^2  & \leq p^t \; \Gamma_c(\rho_{XE}|\rho_E)-1 \nl
p^t \L(\frac{p}{p-1}\R)\Es{a}  \L(\Gamma_c(\rho^a_{ZE}|\rho_E)-\oo{p}\R) -1 \nl
 \L(\frac{p}{p-1}\R) \L[\sum_{a\neq 0} \L(\Gamma_c(\rho^a_{ZE}|\rho_E)-\oo{p}\R)+ \L(\Gamma_c(\rho^0_{ZE}|\rho_E)-\oo{p}\R)\R] -1 \nn \\
& \leq \L(\frac{p}{p-1}\R)\L[ (p^t-1) 2\epsilon \L(  1 -\frac{1}{p} \R) +\L(  1 -\frac{1}{p} \R) \R]-1 \nn \\
&\leq 2 p^t \epsilon\; . 
\ea
On the first line we use Lemma~\ref{lem:2qb},  the second line we use Eq.~\eqref{eq:xorcollision}, and the fourth line we use Lemma~\ref{lem:2qb} and Corollary~\ref{cor:gamma-ub}. Dividing both sides by $4$ and taking the square root, we get the desired result
\ba
\fot \onenorm{\rho_{XE}-U_X\otimes \rho_E}& \leq \frac{p^{t/2}}{\sqrt 2}\sqrt \epsilon\;. 
\ea
\end{proof}

\subsection{Non-uniform XOR lemma}\label{subsection:xor2}
Our non-uniform XOR lemma bounds the distance to uniform of a ccq state, a state with two classical registers and one quantum register. Roughly speaking, the lemma states that given two random variables $X_0 \in \F_p$ and $X \in \F_p^t$, if $X_0+\vev{a,X}$ is close to uniform, then $X_0$ is close to uniform given $X$. 

\newtheorem*{lem:xor2}{Lemma \ref{lem:xor2}}
\begin{lem:xor2}[restated]
Let $p$ be a prime power, $t$ an integer and $\eps\geq 0$. 
Let $\rho_{X_0XE}$ be a ccq state with $X_0 \in \F_p$ and $X=(X_1,\dots,X_t)\in \F^t_p$. For all $a=(a_1,\dots,a_t)\in \F^t_p$, define a random variable $Z=X_0+\vev{a,X}=X_0+\sum_{i=1}^t a_i X_i$. If for all $a$,  $\fot \onenorm{\rho^a_{ZE}-U_Z\otimes \rho_E}\leq \epsilon$, then
\ba
 \fot\big\|\rho_{X_0XE}-U_{X_0}\otimes \rho_{XE}\big\|_1 &\leq \frac{p^{(t+1)/2}}{\sqrt{2}}\sqrt{\epsilon} \;.
\ea
\end{lem:xor2}


The proof of the non-uniform XOR lemma follows the same structure as the proof of the standard XOR lemma in Section~\ref{subsection:xor}: we bound the collision probability by the trace distance in Lemma~\ref{lem:3qb}, then prove the non-uniform XOR lemma based on that. First we establish the analogue of Eq.~\eqref{eq:trace-to-collision} for any ccq state $\rho_{XZE}$:
\ba \label{eq:trace-to-collision2}
&\tr\L((\rho_{XZE}-U_X \otimes \rho_{ZE} )(I_{XZ}\otimes \rho_E^{-1/2})\R)^2 \nn \\
&= \tr\L(\rho_{XZE} \, \rho_E^{-1/2}\R)^2 - 2 \tr \L(\rho_{XZE} \, \rho_E^{-1/2} (U_X \rho_{ZE}) \rho_E^{-1/2}\R) +\tr\L((U_X \rho_{ZE})  \rho_E^{-1/2}\R)^2 \nn \\
&=\Gamma_c(\rho_{XZE}|\rho_E) - \oo{d_X}\Gamma_c(\rho_{ZE}|\rho_E)\;.
\ea
Similar to Lemma~\ref{lem:eig}, we need the following lemma to bound the collision probability by the trace distance in Lemma~\ref{lem:3qb}.  
\begin{lemma}\label{lem:eig2} Let $\rho_{XZE}$ be a ccq state. Then 
\ba
 -\oo{d_X}I_{XZE}\leq &\L(I_{XZ}\otimes \rho_E^{-1/2}\R)(\rho_{XZE}-U_X \otimes \rho_{ZE} )\L(I_{XZ}\otimes \rho_E^{-1/2}\R)   \leq \L(1-\oo{d_X}\R) I_{XZE}\;. 
\ea
\end{lemma}
\begin{proof}
We bound the eigenvalues of the middle expression. Since $\rho_{XZE}$ is a ccq state, we know that the middle expression
\ba
\L(I_{XZ}\otimes \rho_E^{-1/2}\R)(\rho_{XZE}-U_X \otimes \rho_{ZE} )\L(I_{XZ}\otimes \rho_E^{-1/2}\R)=\sum_{x,z} \proj{x}\otimes \proj{z}\otimes \rho_E^{-1/2} \L(\rho^{xz}_E- \oo{d_X} \rho^z_E \R)\rho_E^{-1/2} 
\ea
is block diagonal, where $\rho^z_E=\sum_x\rho^{xz}_E$ and $\rho_E=\sum_{x,z}\rho^{xz}_E$. For any state $\ket{\phi}\in \mH_E$ and $x,z$ in the range of $X,Z$, 
\ba
\vev{ \phi|\rho_E^{-1/2} \L(\rho^{xz}_E- \oo{d_X} \rho^z_E \R)\rho_E^{-1/2}|\phi} \geq \vev{ \phi|\rho_E^{-1/2} \L(- \oo{d_X} \rho^z_E \R)\rho_E^{-1/2}|\phi} \geq -\oo{d_X}\;.
\ea
This proves the first inequality. We also have
\ba
\vev{ \phi|\rho_E^{-1/2} \L(\rho^{xz}_E- \oo{d_X} \rho^z_E \R)\rho_E^{-1/2}|\phi}
&=\vev{ \phi|\rho_E^{-1/2} \L(\rho^{xz}_E- \oo{d_X} \sum_{x'} \rho^{x'z}_E \R)\rho_E^{-1/2}|\phi} \nn \\
&=\L(1- \oo{d_X}\R)  \vev{ \phi|\rho_E^{-1/2} \rho^{xz}_E \rho_E^{-1/2}|\phi} -\oo{d_X} \sum_{x'\neq x}  \vev{ \phi|\rho_E^{-1/2} \rho^{xz}_E \rho_E^{-1/2}|\phi}  \nn  \\
 &\leq   \L(1- \oo{d_X}\R) \;.
\ea
This proves the second inequality.
\end{proof}
We then bound the collision probability by the trace distance as in Lemma~\ref{lem:2qb}.
\begin{lemma}[Bounding collision probability with trace distance, non-uniform]\label{lem:3qb}
Let $ \rho_{XZE}$ be a ccq state. If 
\ba
\fot \onenorm{\rho_{XZE}-U_X \rho_{ZE}}=\epsilon \,, 
\ea 
then 
\ba
\frac{4  \epsilon^2 }{d_X d_Z} \leq \Gamma_c(\rho_{XZE}|\rho_E) - \oo{d_X}\Gamma_c(\rho_{ZE}|\rho_E) \leq 2 \epsilon\L(1- \oo{d_X}\R)\;.
\ea
\end{lemma}
\begin{proof}
The proof is similar to that of Lemma~\ref{lem:2qb}. For the first inequality, we use H\"{o}lder's inequality (Lemma~\ref{lem:holder}) with $r=t=4, \, s=2, \, A=C=I_{XZ} \otimes \rho_E^{1/4}$, and $B=\L(I_{XZ} \otimes \rho_E^{-1/4}\R) \L(\rho_{XZE}-U_X \rho_{ZE}\R) \L(I_{XZ} \otimes \rho_E^{-1/4}\R)$. This leads to
\ba
 2 \epsilon &=  \onenorm{\rho_{XZE}-U_X \rho_{ZE}} \nl
\onenorm{ABC} \nn \\
 & \leq \onenorm{A^4}^{1/4} \onenorm{B^2}^{1/2} \onenorm{C^4}^{1/4} \nl
\sqrt{d_X d_Z\tr\L((\rho_{XZE}-U_X \otimes \rho_{ZE} )\L(I_{XZ}\otimes \rho_E^{-1/2}\R)\R)^2 } \nl
\sqrt{d_X d_Z\L(\Gamma_c(\rho_{XZE}|\rho_E) - \oo{d_X}\Gamma_c(\rho_{ZE}|\rho_E) \R)}\; ,
\ea
where we used Eq.~\eqref{eq:trace-to-collision2} in the last line. Squaring both sides and dividing by $d_X d_Z$, we get the desired inequality.
For the second inequality, we use Lemma~\ref{lem:eig2} to show that 
\ba \label{eq:bound-abs-delta2}
& -\oo{d_X}I_{XZE}\leq \L(I_{XZ}\otimes \rho_E^{-1/2}\R)(\rho_{XZE}-U_X \otimes \rho_{ZE} )\L(I_{XZ}\otimes \rho_E^{-1/2}\R)   \leq \L(1-\oo{d_X}\R) I_{XZE} \nn \\
&\Ra \L|\L(I_{XZ}\otimes \rho_E^{-1/2}\R)(\rho_{XZE}-U_X \otimes \rho_{ZE} )\L(I_{XZ}\otimes \rho_E^{-1/2}\R)\R| \leq \L(1-\oo{d_X}\R) I_{XZE}\;.  
\ea 
Starting with Eq.~\eqref{eq:trace-to-collision2}, we have
\ba
 & \Gamma_c(\rho_{XZE}|\rho_E) - \oo{d_X}\Gamma_c(\rho_{ZE}|\rho_E)\nl
\tr\L((\rho_{XZE}-U_X \otimes \rho_{ZE} )\L(I_{XZ}\otimes \rho_E^{-1/2}\R)\R)^2 \nn \\
 &\leq { \tr \L( \L|\rho_{XZE}-U_X \rho_{ZE}\R|  \L|\L(I_{XZ}\otimes \rho_E^{-1/2}\R)(\rho_{XZE}-U_X \otimes \rho_{ZE} )\L(I_{XZ}\otimes \rho_E^{-1/2}\R) \R|\R)} \nn \\
 &\leq { \tr \L( \L|\rho_{XZE}-U_X \rho_{ZE}\R| \L(1- \oo{d_X}\R) I_{XZE} \R)} \nl
 2 \epsilon  \L(1- \oo{d_X}\R),
\ea
 where we used Eq.~\eqref{eq:bound-abs-delta2} on the fourth line.
\end{proof}
Now we restate and prove the non-uniform XOR lemma. The proof idea is to start from the trace distance of $X_0$ given $X$ to uniform, apply Lemma~\ref{lem:3qb} to get an upper bound in terms of the collision probability of $X_0$ given $X$, apply Eq.~\eqref{eq:trace-to-collision2} and expand the square to express the collision probability of $X_0$ given $X$ in terms of the collision probability of $X_0+\vev{a,X}$, and finally apply Lemma~\ref{lem:3qb} again to get an upper bound in terms of the trace distance of $X_0+\vev{a,X}$ to uniform.

\begin{lem:xor2}[restated]
Let $p$ be a prime power, $t$ an integer and $\eps\geq 0$. 
Let $\rho_{X_0XE}$ be a ccq state with $X_0 \in \F_p$ and $X=(X_1,\dots,X_t)\in \F^t_p$. For all $a=(a_1,\dots,a_t)\in \F^t_p$, define a random variable $Z=X_0+\vev{a,X}=X_0+\sum_{i=1}^t a_i X_i$. If for all $a$,  $\fot \onenorm{\rho^a_{ZE}-U_Z\otimes \rho_E}\leq \epsilon$, then
\ba
 \fot\big\|\rho_{X_0XE}-U_{X_0}\otimes \rho_{XE}\big\|_1 &\leq \frac{p^{(t+1)/2}}{\sqrt{2}}\sqrt{\epsilon} \;.
\ea
\end{lem:xor2}


\begin{proof}
We start by relating the collision probability of $Z$ and $X_0+\vev{a,X}$:
\ba
  \Gamma_c(\rho^a_{ZE}|\rho_E)-\oo{p}&=\tr\L[(\rho^a_{ZE}-U_Z\rho_E)I_Z\otimes \rho_E^{-1/2} \R]^2 \nl
\tr \L[ \sum_z\proj{z} \sum_{x,x_0}\L({\delta}\L(z-x_0-\vev{a,x},0\R)-\oo{p} \R)  \rho^{x_0x}_E I_Z\rho_E^{-1/2}  \R]^2  \nl
\sum_z \tr \L[\sum_{x_0x} \L( {\delta}\L(z-x_0-\vev{a,x},0\R)-\oo{p} \R)  \rho^{x_0x}_E \rho_E^{-1/2}  \R]^2 \nl
\sum_{z,x_0,x_0',x,x'} \L[  {\delta}\L(z-x_0-\vev{a,x},0\R)  {\delta}\L(z-x_0'-\vev{a,x'},0\R)-\frac{2}{p} {\delta}\L(z-x_0-\vev{a,x},0\R)+\frac{1}{p^2} \R] \nn \\
&\,\,\,\,\, \cdot \tr\L(\rho^{x_0x}_E\rho_E^{-1/2}\rho^{x_0'x'}_E\rho_E^{-1/2}\R) \nl
\sum_{x_0,x_0',x,x'} \L[  {\delta}\L(x_0-x_0'+\vev{a,x-x'},0\R)  -\frac{1}{p} \R]\tr\L(\rho^{x_0x}_E\rho_E^{-1/2}\rho^{x_0'x'}_E\rho_E^{-1/2} \R)  \nl 
\sum_{x_0, x_0', x} \L(  {\delta}\L(x_0-x_0',0\R) -\frac{1}{p} \R)\tr\L(\rho^{x_0x}_E\rho_E^{-1/2}\rho^{x_0'x}_E\rho_E^{-1/2} \R) \nn \\ 
&\,\,\,\,\,+ \sum_{x_0,x_0', x \neq x'} \L[  {\delta}\L(x_0-x_0'+\vev{a,x-x'},0\R)  -\frac{1}{p} \R]\tr\L(\rho^{x_0x}_E\rho_E^{-1/2}\rho^{x_0'x'}_E\rho_E^{-1/2} \R) \nl 
\sum_{x_0, x} \tr\L(\rho^{x_0x}_E\rho_E^{-1/2}\rho^{x_0x}_E\rho_E^{-1/2} \R) - \oo{p}\sum_{x_0, x_0', x} \tr\L(\rho^{x_0x}_E\rho_E^{-1/2}\rho^{x_0'x}_E\rho_E^{-1/2} \R)  \nn \\  
&\,\,\,\,\,+\sum_{x_0,x_0', x \neq x'} \L[  {\delta}\L(x_0-x_0'+\vev{a,x-x'},0\R)  -\frac{1}{p} \R]\tr\L(\rho^{x_0x}_E\rho_E^{-1/2}\rho^{x_0'x'}_E\rho_E^{-1/2} \R) \nl 
 \Gamma_c(\rho_{X_0XE}|\rho_E) - \oo{p}\Gamma_c(\rho_{XE}|\rho_E) \nn \\
&\,\,\,\,\,+ \sum_{x_0,x_0', x \neq x'} \L[  {\delta}\L(x_0-x_0'+\vev{a,x-x'},0\R)  -\frac{1}{p} \R]\tr\L(\rho^{x_0x}_E\rho_E^{-1/2}\rho^{x_0'x'}_E\rho_E^{-1/2} \R) .
\ea 
When we average over $a$, the last term vanishes,
\ba \label{eq:xor2collision}
\Es{a}\L( \Gamma_c(\rho^a_{ZE}|\rho_E)-\oo{p} \R)  = \Gamma_c(\rho_{X_0XE}|\rho_E) - \oo{p}\Gamma_c(\rho_{XE}|\rho_E) \;.
\ea
With the heavy work done, we put everything together and prove the lemma
\ba
\frac{\onenorm{\rho_{X_0XE}-U_{X_0}\rho_{XE}}^2}{p^{t+1}} &\leq  \Gamma_c(\rho_{X_0XE}|\rho_E) - \oo{p}\Gamma_c(\rho_{XE}|\rho_E) \nl
\Es{a}\L( \Gamma_c(\rho^a_{ZE}|\rho_E)-\oo{p} \R)  \nn \\
&\leq 2\epsilon \, ,
\ea
where we used Lemma~\ref{lem:3qb} one the first line, Eq.~\eqref{eq:xor2collision} on the second line, Lemma~\ref{lem:3qb} and the assumption of the lemma on the third line. Multiplying both sides by $\frac{p^{t+1}}{2}$ and take a square root, we get the desired result:
\ba
\fot\onenorm{\rho_{X_0XE}-U_{X_0}\rho_{XE}} &\leq \frac{p^{(t+1)/2}}{\sqrt{2}}\sqrt{\epsilon}\;.
\ea
\end{proof}


\section{Quantum-Proof Non-malleable Extractor}\label{sec:nmext}

In this section we introduce our non-malleable extractor and prove its security. The extractor was first considered by Li~\cite{Li12a}. 
We use the symbol $\|$ for concatenation of strings, and for $a,b \in \F_p^n$ 
 write  $\vev{a,b}$ for the standard inner product over $\F_p^n$.

\begin{definition}[Inner product-based non-malleable extractor]\label{def:ip_nmext}
Let $p\neq 2$ be a prime.  For any even integer $n$, define a function 
 $\nmExt : \F_p^n \times \F_p^{n/2} \to \F_p$ by 
$\nmExt(X,Y) =\vev{X, Y||Y^2}$, where $Y^2$ is defined as in Section~\ref{sec:notation}. 
\end{definition}

\newtheorem*{thm:main}{Theorem \ref{thm:main}}
\begin{thm:main}
Let $p\neq 2$ be a prime. Let $n$ be an even integer. Then for any $\eps>0$ the function $\nmExt(X,Y)=\langle X, Y\|Y^2 \rangle$ is an $(\L(\frac{n}{2}+6\R){\log p} -1+4\log\oo{\eps},\eps)$  quantum-proof non-malleable extractor. 
\end{thm:main}

The proof of Theorem~\ref{thm:main} is based on a reduction showing that any successful attack for an adversary to $\nmExt$ leads to a good strategy for the players in a certain communication game, that we introduce next. 

\subsection{A communication game}
\label{sec:com-game}

Let $p\neq 2$ be a prime. Let $n$ be an even integer, and $g:\F_{p}^{n/2} \times \F_{p}^{n/2} \to \F_p^n$ an arbitrary function such that for any $z\in \F_p^n$ there are at most two possible pairs $(y,y')$ such that $y\neq y'$ and $g(y,y')=z$. 
Consider the following communication game, called $\guess(n,p,g)$, between two players Alice and Bob.

\begin{enumerate}
\item Bob receives $y\in \F_p^{n/2}$.
\item Alice creates a cq state $\rho_{XE}$, where $X\in\F_p^n$, and sends the quantum register $E$ to Bob. 
\item Bob returns $y'\in \F_p^{n/2}$ and $b\in \F_p$. 
\end{enumerate}
The players win if and only if $b= \vev{x ,g(y,y')}$ and $y' \neq  y$. 

The following lemma bounds the players' maximum success probability in this game as a function of the min-entropy of Alice's input $X$, conditioned on her message to Bob.

\begin{lemma}[Success probability of the communication game]\label{lem:commugame}
Suppose there exists a communication protocol for Alice and Bob in $\guess(n,p,g)$  that succeeds with probability at least $\frac{1}{p}+\eps$, on average over a uniformly random choice of input $y$ to Bob. Then  $\Hmin(X|E)_\rho \leq \frac{n}{2}\log p+1+2\log\frac{1}{\eps}$.
\end{lemma}

\begin{proof}  
Let $\rho_{XE} = \sum_x \ket{x}\!\bra{x}_X \otimes \rho_E^x$ be the cq state prepared by Alice. A strategy for Bob is a family of POVM $\{M_{y}^{y',b}\}_{y',b}$, indexed by $y\in \F_{p}^{n/2}$ and with outcomes $(y',b) \in \F_{p}^{n/2}\times \F_p$. We can assume that $\{M_{y}^{y',b}\}_{y',b}$ is projective, since Alice can send ancilla qubits along with $\rho$ and allow Bob to apply Naimark's theorem to his POVM in order to obtain a projective measurement; this will change neither his success probability nor the min-entropy of Alice's state. By definition, the players' success probability in $\guess(n,p,g)$ is 
\begin{equation}\label{eq:ps-0}
\frac{1}{p} + \eps \,=\, \sum_x \,p^{-\frac{n}{2}} \sum_{y} \sum_{y'} \sum_b\, {\delta}({b, \langle x, g(y,y')\rangle})\, \Tr\big( M_{y}^{y',b} \,\rho_E^x\big)\;.
\end{equation}
For each $u\in \F_p$ let $A_{y,u}^{y'} = \sum_b \omega^{ub} M_{y}^{y',b}$, where $\omega = e^{\frac{2i\pi}{p}}$. By inversion, $M_{y}^{y',b} = \frac{1}{p} \sum_u \omega^{-ub} A_{y,u}^{y'}$. Replacing this into~\eqref{eq:ps-0} we obtain
\begin{align}
\frac{1}{p} + \eps & =\frac{1}{p}\sum_u \,p^{-\frac{n}{2}} \sum_{y} \sum_{y'} \sum_b\, {\delta}({b, \langle x, g(y,y')\rangle})\, \omega^{-ub} \,\Tr\big( A_{y,u}^{y'} \,\rho_E^x\big)\notag\\
&\leq \frac{1}{p} + \Big(1-\frac{1}{p}\Big)\max_{u\neq 0}  \Big|p^{-\frac{n}{2}} \sum_{y} \sum_{y'} \sum_b\, {\delta}({b, \langle x, g(y,y')\rangle})\, \omega^{-ub} \,\Tr\big( A_{y,u}^{y'} \,\rho_E^x\big)\Big| \;,\label{eq:ps-1}
\end{align}
where for the second line we used that $\sum_{y'} A_{y,0}^{y'} = \sum_{y',b} M_y^{y',b} = I_E$. 

Fix $u\neq 0$ that achieves the maximum in~\eqref{eq:ps-1}. For fixed $y$, define the  map $T_{y,u}$ on $\mH_E$ by
\ba
T_{y,u}: \ket{\psi} \mapsto \sum_{y'} \ket{y'} A_{y,u}^{y'} \ket{\psi}\;.
\ea 
$T_{y,u}$ has norm at most $1$, since 
$$T_{y,u}^\dag T_{y,u} =\sum_{y'} (A_{y,u}^{y'})^\dagger  A_{y,u}^{y'}= \sum_{y'} \,\sum_b \Big(\,M_{y}^{y',b} \Big)^2 = I_E\;.$$
For the second equality we used that  $\{M_{y}^{y',b}\}_{y',b}$ is projective. Therefore we can complete $T_{y,u}$ into a unitary map by adding arbitrary, unused outcomes; for simplicity we'll assume $T_{y,u}$ itself is unitary. 

Consider the following guessing strategy for an adversary holding side information $\rho_E^x$ about $x$. The adversary first prepares a uniform superposition over $y$. Conditioned on $y$, it applies the map $T_{y,u}$. It computes $g(y,y')$ in an ancilla register, and erases $(y,y')$, except for one bit of information $r(y,y')\in\{0,1\}$, which specifies which pre-image $(y,y')$ is, given $ g(y,y')$ (this is possible by the $2$-to-$1$ assumption on $g$). The adversary applies a Fourier transform on the register containing $g(y,y')$, using $\omega_u = \omega^{-u}$ as primitive $p$-th root of unity (this is possible since $u\neq 0$ and $p$ is prime). It measures the result and outputs it as a guess for $x$. Formally, the transformation this implements is
\begin{align*}
\ket{\psi}  &\mapsto p^{-\frac{n}{4}} \sum_y \ket{y} \sum_{y'} \ket{y'} A_{y,u}^{y'} \ket{\psi}\\
&\mapsto  p^{-\frac{n}{4}} \sum_{y,y'}  \ket{g(y,y')} \ket{r(y,y')} A_{y,u}^{y'} \ket{\psi}\\
&\mapsto   \sum_v \ket{v} \Big(p^{-\frac{3n}{4}} \sum_{y,y'} \omega_u^{\langle v, g(y,y')\rangle} \ket{r(y,y')} A_{y,u}^{y'} \Big)\ket{\psi}\;.
\end{align*}
The adversary's success probability in guessing $v=x$ on input $\rho_E^x$ is therefore 
\begin{align}
p_s &= \sum_x \,\Tr\Big( \Big(p^{-\frac{3n}{4}} \sum_{y,y'} \omega_u^{\langle x, g(y,y')\rangle} \ket{r(y,y')}\otimes A_{y,u}^{y'} \Big) \rho_E^x \Big(p^{-\frac{3n}{4}} \sum_{y,y'} \omega_u^{-\langle x, g(y,y')\rangle} \bra{r(y,y')}\otimes  (A_{y,u}^{y'})^\dagger \Big)\Big)\notag\\
&= \frac{1}{p^{\frac{3n}{2}}} \sum_x\, \sum_{r\in\{0,1\}} \Tr\Big( \Big(\sum_{y,y':\, r(y,y')=r} \omega_u^{\langle x, g(y,y')\rangle} A_{y,u}^{y'} \Big)^\dagger\Big(\sum_{y,y':\, r(y,y')=r} \omega_u^{\langle x, g(y,y')\rangle} A_{y,u}^{y'} \Big) \rho_E^x \Big)\notag\\
&\geq \frac{1}{p^{\frac{3n}{2}}} \sum_x \, \frac{1}{2}\,\Tr\Big( \Big(\sum_{y,y'} \omega_u^{\langle x, g(y,y')\rangle} A_{y,u}^{y'} \Big)^\dagger \Big(\sum_{y,y'} \omega_u^{\langle x, g(y,y')\rangle} A_{y,u}^{y'} \Big) \rho_E^x \Big)\; ,\label{eq:ps-2}
\end{align}
where for the last line we used $a^2 + b^2 \geq \frac{1}{2} (a+b)^2$. Now, recall from~\eqref{eq:ps-1} and our choice of $u$ that 
\begin{align}
\eps &\leq \,p^{-\frac{n}{2}} \Big|\sum_{x,y,y'} \, \omega^{-u(\langle x, g(y,y')\rangle)} \,\Tr\big( A_{y,u}^{y'} \,\rho_E^x\big)\Big|\notag\\
&\leq p^{-\frac{n}{2}}\Big(\sum_x \Tr(\rho_E^x)\Big)^{1/2}\Big( \sum_x\, \Tr\Big(  \Big(\sum_{y,y'} \, \omega^{-u(\langle x, g(y,y')\rangle)} \, A_{y,u}^{y'}\Big) \,\rho_E^x\,\Big(\sum_{y,y'} \, \omega^{-u(\langle x, g(y,y')\rangle)} \, A_{y,u}^{y'}\Big)^\dagger\Big)\Big)^{1/2}\;,\label{eq:ps-3}
\end{align}
where the inequality is Cauchy-Schwarz.  Comparing~\eqref{eq:ps-2} and~\eqref{eq:ps-3} gives 
$$ p_s \,\geq\, \frac{1}{2} p^{-\frac{n}{2}} \eps^2\;.$$
We conclude using that by Lemma~\ref{lem:minentropyguess}, $\Hmin(X|E) \leq -\log p_s$. 
\end{proof}

\subsection{Proof of Theorem~\ref{thm:main}}

In this section we give the proof of Theorem~\ref{thm:main}. Towards this we first prove a preliminary lemma showing that a certain function, based on the definition of $\nmExt$, has few collisions. 

\begin{lemma}\label{lem:221}
Let $p \neq 2$ be a prime and $n$ an even integer. For $a \in \F_{p}$ define a function $g_a:\F_{p}^{n/2} \times \F_{p}^{n/2} \to \F_{p}^n$ by 
\ba\label{eq:def-ga}
g_a(y,y') \,=\, y+ay'\|y^2+ay'^2\;,
\ea
where $y^2$ is defined in Section~\ref{sec:notation}. Then for any  $a \in \F_{p}, a\neq 0$ and $z\in\F_p^n$ there are at most $2$ distinct pairs $(y,y')$ such that $y'\neq  y$ and $g_a(y,y')=z$.
\end{lemma}

\begin{proof}
We use the bijection defined in Section~\ref{sec:notation} to interpret $y$ and $y'$  in $\F_{p^{n/2}}$. For $a\neq0$, we fix an image $g_a=(c,d)$, where $c,d$ are interpreted as elements of $\F_{p^{n/2}}$, and solve for $(y,y')$ in $\F_{p^{n/2}}\times \F_{p^{n/2}}$ satisfying
\ba  \label{eq:g1storder}
 y+ay' &=c \;,\\ 
 y^2+ay'^2 &=d\;.
\ea
 Using~(\ref{eq:g1storder}) to eliminate $y$ we get
\ba
&(c-ay')^2+ay'^2=d \nn \\
&\Ra (a+a^2)y'^2+(-2ca)y'+(c^2-d)=0 \; . \label{eq:g1storder-2}
\ea
Since~\eqref{eq:g1storder-2} is a quadratic equation, there are at most two solutions unless all coefficients are zero. Since $p\neq2$, $-2\neq 0$. If all coefficients are zero, $-2\neq 0$, and $a \neq0$,  then $c=d=0,a=-1$, which implies $y' =  y$ by~(\ref{eq:g1storder}) and contradicts our assumption. So there are at most two different $y'$ that can be mapped to $(c,d)$. By~(\ref{eq:g1storder}) each $y'$ corresponds to a unique $y$, so there are at most two pre-images.
\end{proof}

We are ready to give the proof of Theorem~\ref{thm:main}. The proof depends on a simple lemma relating trace distance and guessing measurements, Lemma~\ref{lem:dist2guess}, which is stated and proved after the proof of the theorem.

\begin{proof}[Proof of Theorem~\ref{thm:main}]
Let $k=\L(\frac{n}{2}+6\R)\log p-1+4\log\oo{\eps}$ and 
$\rho_{XE} \in \Density(\C^{p^{n
}} \otimes \mH_E)$ an $(n
\log p,k)$-source. Fix a CPTP map $\Adv: \Lin(\C^{p^{n/2}}\otimes \mH_E) \to \Lin(\C^{p^{n/2}} \otimes \mH_{E'})$ with no fixed points, and define $\sigma_{\nmExt(X,Y) \nmExt(X,Y') YY' E'}$ as in Definition~\ref{def:nmext}. Given the definition of $\nmExt$, to prove the theorem we need to show that 
\ba\label{eq:nmext-last}
(\vev{X,Y\|Y^2},\vev{X,Y'\|Y'^2},Y',Y,E')_\sigma \approx_\eps (U_{\F_p},\vev{X,Y'\|Y'^2},Y',Y,E')_\sigma \;.
\ea 
Applying the XOR lemma, Lemma~\ref{lem:xor2}, with $X_0 =\vev{X,Y||Y^2}$, $X = \vev{X,Y'||Y'^2}$, $E = (Y',Y,E')$ and $t=1$,~\eqref{eq:nmext-last} will follow once it is shown that 
\ba\label{eq:nm-final-1}
&(\vev{X,Y||Y^2}+a \vev{X,Y'\|Y'^2},Y',Y,E')_\sigma \approx_{\frac{2\eps^2}{p^2}} (U_{\F_p},Y',Y,E')_\sigma \;, 
\ea
 for all $a\in \F_p$. For $a=0$, ~\eqref{eq:nm-final-1} follows from the fact that inner product is a quantum-proof two source extractor, which can be shown by the combination of  Theorem 5.3 of \cite{clw14} and Lemma~1 in~\cite{LLST05}. 
 For non-zero $a\in \F_p$, recall the function $g_a: \F_{p}^{n/2} \times \F_{p}^{n/2} \to \F_p^n$ defined in~\eqref{eq:def-ga}. 
Lemma~\ref{lem:221} shows that for any $a\neq 0$, the restriction of $g_a$ to $\{(y,y'):\,y\neq y'\}$ is at most $2$-to-$1$, which is ensured by the fact that $\Adv$ has no fixed points. 
We establish~\eqref{eq:nm-final-1} by contradiction. Assume thus that 
\ba \label{eq:beforecommu}
(\vev{X,g_a(Y,Y')},Y',Y,E' )_\sigma \approx_{\frac{2\eps^2}{p^2}}( U_{\F_p},Y',Y,E' )_\sigma
\ea
does not hold, for some non-zero $a\in \F_p$. Fix such an $a$ and write $g_a$ for $g$. 
From Lemma~\ref{lem:dist2guess} it follows that there exists a POVM measurement $\{M^z\}_{z\in\F_p}$ on $\sigma_{Y'YE'}$ such that 
\ba\label{eq:nm-final-2}
 \sum_{z\in \F_p}\, \Tr\big( M^z \sigma^z_{YY'E} \big) \,\geq\,\oo{p}+ \frac{2\eps^2}{p^3}\;,
\ea
where $\sigma^z_{YY'E}$ is the reduced density of $\sigma$ on $YY'E$ conditioned on $\vev{X,g(Y,Y')}=z$. To conclude the proof of the theorem we show that the adversary's map $\Adv$ and the POVM $\{M^z\}$ can be combined to give a ``successful'' strategy for the players in the communication game introduced in Section~\ref{sec:com-game}. To see this, consider the state $\rho_{XE}$ that is instantiated as the source for the extractor; by definition $\Hmin(X|E)_\rho = k = \L(\frac{n}{2}+6\R)\log p-1+4\log\oo{\eps}$. In the third step of the game, Bob applies the map $\Adv$ to the registers $Y$ and $E$ containing his input $Y$ and the state sent by Alice, and measures to obtain an outcome $Y'$. He then applies the measurement $\{M^z\}$ on his registers $(Y,Y',E)$ to obtain a value $b=z\in \F_p$ that he provides as his output in the game. By~\eqref{eq:nm-final-2} it follows that this strategy succeeds in the game with probability at least $\oo{p}+ \frac{2\eps^2}{p^3}$, which by Lemma~\ref{lem:commugame} implies  $\Hmin(X|E) \leq \frac{n}{2}\log p+1+2\log\frac{p^3}{2\eps^2} $, contradicting our choice of $k$. This proves~\eqref{eq:nm-final-1} and thus the theorem. 
\end{proof}

The following lemma is used in the proof of the theorem. 

\begin{lemma} \label{lem:dist2guess}
Let $\rho_{XE} = \sum_{x} \ket{x}\!\bra{x}\otimes \rho_E^x$ be such that 
$$\fot \|(X,E)-(U,E)\|_1\,=\,\fot \big\|\rho_{XE}- U_X \otimes \rho_E \big\|_1 \,= \, \eps\;,$$
 where $U_X$ is the totally mixed state on $X$ and $\rho_E = \sum_x \rho_E^x$. Then there exists a POVM $\{M_x\}$ on $\rho_E$ such that 
$$\sum_{x} \Tr(M_x\rho_E^x) = \frac{1}{d_X} + \frac{\eps}{d_X}\;.$$
\end{lemma}

\begin{proof}
Since $\rho_{XE}$ is a cq state,  $\|\rho_{XE}- U_X \otimes \rho_E\|_1 = \sum_x \|\rho_E^x - \frac{1}{d_X}\rho_E\|_1$. For each $x$, let $M'_x $ be the projector onto the positive eigenvalues of $ \|\rho_E^x - \frac{1}{d_X}\rho_E\|_1$, so 
\begin{equation}\label{eq:norm1-1}
\sum_x \Tr( M'_x(\rho_E^x - \oo{d_X}\rho_E)) = \fot\sum_x \|\rho_E^x -\oo{d_X}\rho_E\|_1\;.
\end{equation}
Let $M' = \sum_x M'_x$ and $M_x =\oo{d_X}( M'_x + (I_E-\oo{d_X} M'))$. Then $M_x\geq 0$ and $\sum_x M_x =\oo{d_X}(M'+d_XI_E-M') = I_E$. Moreover,
\begin{align*}
\sum_x \Tr(M_x\rho_E^x) &=\sum_x \Tr\L[ \oo{d_X}( M'_x + (I_E-\oo{d_X} M')) \rho_E^x \R] \\
&= \frac{1}{d_X} \L[\sum_x \L(\Tr( M'_x \rho_E^x )\R)+ \Tr\L((I_E-\oo{d_X} M')  \rho_E \R)\R] \\
&= \frac{1}{d_X} +\frac{1}{d_X} \sum_x \Big(\Tr(M'_x \rho_E^x)  - \frac{1}{d_X}\Tr(M'_x \rho_E) \Big)\\
&= \frac{1}{d_X} + \frac{1}{d_X} \Big(\sum_x \Tr\big( M'_x (\rho_E^x - \frac{1}{d_X}\rho_E)\big)\Big)\\
&= \frac{1}{d_X} + \frac{1}{2d_X} \sum_x \Big\|\rho_E^x - \frac{1}{d_X}\rho_E\Big\|_1
\end{align*}
by~\eqref{eq:norm1-1}.
\end{proof}

\section{Privacy amplification}
\label{sec:pa}

Dodis and Wichs~\cite{DW09} introduced a framework for constructing a two-message  privacy amplification protocol from any non-malleable extractor. In~\cite{CV16} it is shown that the same framework, when instantiated with a quantum-proof non-malleable extractor $\nmExt$ as defined in Definition~\ref{def:nmext}, leads to a protocol that is secure against active quantum adversaries. In Section~\ref{sec:dw-nm} we recall the Dodis-Wichs protocol, and state the security guarantees that follow by plugging in our non-malleable extractor construction. The guarantees follows from the quantum extension of the Dodis-Wichs results in~\cite{CV16}; since that work has not been published we include their results regarding the Dodis-Wichs protocol in Appendix~\ref{app:pa}.

In Section~\ref{one-round} we show that a different protocol for privacy amplification due to Dodis et al.~\cite{DKKRS12}, whose main advantage is of being a one-round protocol, is also quantum-proof. The construction and analysis of the protocol of~\cite{DKKRS12} is simple, with the drawback of a large entropy loss.

We start with the definition of a quantum-secure privacy amplification protocol against active adversaries. A privacy amplification protocol $(P_A, P_B)$ is defined as follows. The protocol is executed by two parties Alice and Bob sharing a
secret $X\in \{0,1\}^n$, whose actions are described by $P_A$, $P_B$ respectively.\footnote{It is not necessary for the definition to specify exactly how the protocols are formulated; informally, each player's actions is described by a sequence of efficient algorithms that compute the player's next message, given the past interaction.} In addition there is an active, computationally
unbounded adversary Eve, who might have some quantum side information $E$
correlated with $X$ but satisfying $\Hmin(X|E)_{\rho} \ge k$, where $\rho_{XE}$ denotes the initial state at beginning of the protocol. 

Informally, the goal for the protocol is that
whenever a party (Alice or Bob) does not reject, the key $R$ output by
this party is random and statistically independent of Eve's
view. Moreover, if both parties do not reject, they must output the
same keys $R_A=R_B$ with overwhelming probability.

More formally, we assume that Eve is in full control of the
communication channel between Alice and Bob, and can arbitrarily
insert, delete, reorder or modify messages sent by Alice and Bob to
each other. 
At the end of the
protocol, Alice outputs a key $R_A\in
\{0,1\}^m \cup \{\perp\}$, where $\perp$ is a special symbol indicating
rejection. Similarly, Bob outputs a key $R_B \in \{0,1\}^m \cup
\{\perp\}$. The following definition generalizes the classical definition in \cite{DLWZ11}.

\begin{definition}\label{privamp}
Let $k,m$ be integer and $\eps\geq 0$. A privacy amplification protocol $(P_A, P_B)$ is
a $(k, m, \eps)$-\emph{privacy amplification protocol secure against active quantum adversaries} if it satisfies
the following properties for any initial state $\rho_{XE}$ such that $\Hmin(X|E)_\rho \geq k$, and where  $\sigma$ be the joint state of Alice, Bob, and Eve at the end of the protocol:
\begin{enumerate}
\item \emph{Correctness.} If the adversary does not interfere with the protocol, then $\Pr[R_A=R_B \land~ R_A\neq \perp \land~ R_B\neq \perp]=1$.
\item \emph{Robustness.} This property comes in two flavors. The first is {\em pre-application} robustness, which states that even in the presence of an active adversary, $\Pr[R_A \neq R_B \land~ R_A \neq \perp \land~ R_B \neq \perp]\le \eps$. The second is {\em post-application} robustness, which is defined similarly, except the adversary is additionally given the key $R_A$ that is the result of the interaction $(P_A,P_E)$, and the key $R_B$ that results from the interaction $(P_E,P_B)$, where $P_E$ denotes the adversary's actions in its interaction with Alice and Bob. 

\item \emph{Extraction.} Given a string $r\in \{0,1\}^m\cup \{\perp\}$, let $\mathsf{purify}(r)$ be a random variable on $m$-bit strings that is deterministically equal to $\perp$ if $r=\perp$, and is otherwise uniformly distributed. Let $V$ denotes the transcript of an execution of the protocol execution, and $\rho_{E'}$ the final quantum state possessed by Eve. Then the following should hold: 
\[(R_A, V, E')_\sigma \approx_{\eps} (\mathsf{purify}(R_A), V, E')_\sigma
~~~~\mbox{and}~~~~
(R_B, V, E')_\sigma \approx_{\eps} (\mathsf{purify}(R_B), V, E')_\sigma\;.\]
In other words, whenever a party does not reject, the party's key is indistinguishable from a fresh random string to the adversary.
%
%
\end{enumerate}
The quantity $k-m$ is called the \emph{entropy loss}. 
\end{definition}

\subsection{Dodis-Wichs protocol with non-malleable extractor}
\label{sec:dw-nm}

Here we first recall the Dodis-Wichs protocol for privacy amplification (hereafter called \emph{Protocol DW}), which is summarized in Figure~\ref{fig:pa}, and the required security definitions, taken from~\cite{CV16}. We then state the result obtained by instantiating the protocol with the quantum-proof non-malleable extractor from Theorem~\ref{thm:main}. 

\begin{figure}
\begin{protocol*}{Protocol DW}
\begin{description}
\item Let $d_X,d_Y,d_2,\ell,d_Z,t,k$ be integers and $\eps_\mac,\eps_\Ext,\eps_\nmExt>0$.
\item Let $\mac:\{0,\ldots,d_Z-1\} \times \{0,1\}^{d_2} \to \{0,1\}^t$ be a one-time $\eps_\mac$-information-theoretically secure message authentication code.
\item Let $\Ext:\{0,\ldots,d_X-1\}\times\{0,1\}^{d_2} \to \{0,1\}^m$ be a strong $(k-\ell-\log(1/\eps_{\Ext}),\eps_\Ext)$  quantum-proof extractor.
\item Let $\nmExt:\{0,\ldots,d_X-1\}\times\{0,,\ldots,d_Y-1\} \to \{0,\ldots,d_Z-1\}$ be a $(k,\eps_\nmExt)$ quantum-proof non-malleable extractor.
\item  It is assumed that both parties, Alice and Bob, have access to a shared random variable $X\in\{0,\ldots,d_X-1\}$.
\begin{step}
\item[1.] Alice samples a $Y_A$ uniformly from $\{0,,\ldots,d_Y-1\}$. She sends $Y_A$ to Bob. She computes $Z = \nmExt(X,Y_A)$.
\item[2.] Bob receives $Y'_A$ from Alice. He samples a uniform $Y_B \sim U_{{d_2}}$, and computes $Z' = \nmExt(X,Y'_A)$ and $W = \mac(Z',Y_B)$. He sends $(Y_B,W)$ to Alice. Bob then reaches the $\kd$ state and outputs $R_B = \Ext(X,Y_B)$.
\item[3.] Alice receives $(Y'_B,W')$ from Bob. If $W' = \mac(Z,Y'_B)$ she reaches the $\kc$ state and outputs $R_A = \Ext(X,Y_B')$. Otherwise she outputs $R_A=\perp$.
\end{step}
\end{description}
\end{protocol*}
\caption{The Dodis-Wichs privacy amplification protocol.}
\label{fig:pa}
\end{figure}

Aside from the use of a strong quantum-proof extractor (Definition~\ref{def:ext}) and a quantum-proof non-malleable extractor (Definition~\ref{def:nmext}), the protocol relies on an information-theoretically secure one-time message authentication codes, or $\mac$. This security notion is defined as follows.

\begin{definition}\label{def:mac}
A function $\mac:\{0,\ldots,d_Z-1\}\times\{0,1\}^d \to \{0,1\}^t$ is an \emph{$\eps_\mac$-information-theoretically secure one-time message authentication code} if for any function $\mathcal{A}:\{0,1\}^d \times \{0,1\}^t \to \{0,1\}^d\times \{0,1\}^t$ it holds that for all $m\in \zo^d$
$$\Pr_{k\leftarrow {U_{Z}}}\big[ (\mac(k,m') = \sigma' ) \, \wedge \, (m'\neq m) : (m',\sigma') \leftarrow \mathcal{A}(m,\mac(k,m))\big] \,\leq\,\eps_\mac.$$  
\end{definition}
Efficient constructions of $\mac$ satisfying the conditions of Definition~\ref{def:mac} are known. The following proposition summarizes some parameters that are achievable using a construction based on polynomial evaluation.

\begin{proposition}[Proposition 1 in~\cite{renner2005universally}]\label{prop:mac}
For any $\eps_\mac>0$, integer $d > 0$, $d_Z \ge \frac{d^2}{\eps_\mac^2}$, there exists an efficient family of  $\eps_\mac$-information-theoretically secure one-time message authentication codes $$\{\mac:\{0,\ldots,d_Z-1\} \times \{0,1\}^d \to \{0,1\}^t\}_{d\in\N}$$ with $t\leq \log d + \log(1/\eps_\mac)$.
\end{proposition}


The correctness and security requirements for the protocol are natural extensions of the classical case (see Definition 18 in~\cite{DW09}). Informally, the adversary has the following control over the outcome of the protocol. First, it possess initial quantum side information $E$ about the weak secret $X$ shared by Alice and Bob. That is, it has a choice of a cq source $\rho_{XE}$, under the condition that $\Hmin(X|E)$ is sufficiently large. Second, the adversary may intercept and modify any of the messages exchanged. In Protocol DW there are only two messages exchanged, $Y_A$ from Alice to Bob and $(Y_B,\sigma)$ from Bob to Alice. To each of these messages the adversary may apply an arbitrary transformation, that may depend on its side information $E$. We model the two possible attacks, one for each message, as CPTP maps $T_1 : \Lin(\mH_Y \otimes \mH_E) \to \Lin(\mH_Y \otimes \mH_{E'})$ and $T_2 : \Lin(\C^{2^{d_2}} \otimes \mH_{2^{t}}\otimes  \mH_{E'}) \to \Lin(\C^{2^{d_2}} \otimes \C^{2^{t}} \otimes \mH_{E''})$, where $\mH$ denotes the Hilbert space associated with system $E$. Note that we may always assume that $\mH$ is large enough for the adversary to keep a local copy of the messages it sees, if it so desires. 





The following result on the security of protocol DW is shown in~\cite{CV16}. We include the proof in Appendix~\ref{app:pa}. 

\begin{theorem}\label{thm:pa-sec}
Let $k,t,d_Z$ and $\eps_\mac,\eps_\Ext,\eps_\nmExt$ be parameters of Protocol DW, as specified in Figure~\ref{fig:pa}. Let $\nmExt$ be a $(k,\eps_\nmExt)$ quantum-proof non-malleable extractor, $\Ext$ a strong $(k-\log d_Z-\log(1/\eps_\Ext),\eps_\Ext)$ quantum-proof extractor, and $\mac$ an $\eps_\mac$-information-theoretically secure one-time message authentication code. Then for any active attack $(\rho_{XE},T_1,T_2)$ such that $\Hmin(X|E)_\rho \geq k$, the DW privacy amplification protocol described in Figure~\ref{fig:pa} is $(k,m,\eps)$-secure as defined in Definition~\ref{privamp} with $\eps = O(\eps_\Ext + \eps_{\nmExt} + \eps_\mac)$.
\end{theorem}

Combined with Theorem~\ref{thm:main} stating the security of our construction of a quantum-proof non-malleable extractor, Theorem~\ref{thm:pa-sec} provides a means to obtain privacy amplification protocol secure against active attacks for a range of parameters. Due to the limitations of our non-malleable extractor we are only able to extract from sources whose entropy rate is at least $\frac{1}{2}$. This is a typical setting in the case of quantum key distribution, where the initial min-entropy satisfies $\Hmin(X|E) \geq \alpha \log d_X$ for some constant $\alpha$ which depends on the protocol and the noise tolerance, but is generally larger than $3/4$. Specifically, we obtain the following:

\begin{corollary}\label{cor:dw-pa} For any $\eps > 0$, there exists a constant $c > 0$, such that the following holds. For any active attack $(\rho_{XE},T_1,T_2)$ such that $\Hmin(X|E)_\rho  = k \geq \frac{1}{2} \log d_X + c \cdot \logeps$, there is an $O(\eps)$-secure DW protocol that outputs a key of length $m = k - O(\logeps)$. 
\end{corollary}
\begin{proof}
Let $p$ be a prime and $n$ a positive integer such that $\log p = \Theta(\logeps)$ and $d_X = p^n$. Let $d_Y = p^{n/2}$, and $d_Z = p$. Also, let $d_2 = O(\log d_X)$, $m = k - O(\logeps)$, and $t = O(\logeps)$. We instantiate Theorem~\ref{thm:pa-sec} with the following.
\begin{itemize}
\item Let $\Ext:\{0,\ldots,d_X-1\}\times\{0,1\}^{d_2} \to \{0,1\}^m$ be the $(k - O(\log(1/\eps)),\eps)$ strong quantum-proof extractor from Theorem~\ref{thm:quantum-LHL}.
\item  Let $\nmExt:\{0,\ldots,d_X-1\}\times\{0,,\ldots,d_Y-1\} \to \{0,\ldots,d_Z-1\}$ be the $(\frac{1}{2} \cdot \log d_X + O(\logeps),\eps)$ non-malleable extractor from Theorem~\ref{thm:main}.
\item Let $\mac:\{0,\ldots,d_Z-1\} \times \{0,1\}^{d_2} \to \{0,1\}^t$ be the one-time $\eps$-information-theoretically secure message authentication code from Proposition~\ref{prop:mac}.
\end{itemize}
The result follows.
\end{proof}

\subsection{One-round Privacy Amplification Protocol} \label{one-round}
In this section we show that the one-round protocol of Dodis et al.~\cite{DKKRS12} is also quantum-proof. This protocol has significantly higher entropy loss, $(n/2) + \log(1/\eps)$, than the DW protocol we presented in the previous section.

\begin{figure}
\begin{protocol*}{One-round Privacy Amplification Protocol }
\begin{description}
\item Let $n,k$ be integers and $\eps > 0$. Let $v = n - k + \log(1/\eps)$ and $m = (n/2)-v$.
\item  It is assumed that both parties, Alice and Bob, have access to a shared random variable $X\in \{0,1\}^n$. They interpret $X$ as a pair $X= (X_1, X_2)$ where $X_1, X_2$ are identified as elements in $\F_{2^{n/2}}$.
\begin{step}
\item[1.] Alice samples a $Y$ uniformly from $\F_{2^{n/2}}$ and computes $Z = Y X_1 + X_2$. Let $W = [Z]_1^v$ be the first $v$ bits of $Z$. She sends $(Y,W)$ to Bob and outputs $R_A = [Z]_{v+1}^{n/2}$, the remaining part of $Z$.
\item[2.] Bob receives $(Y', W')$ from Alice and computes $Z' = Y' X_1 + X_2$. If $W' = [Z']_1^v$, then Bob outputs $R_B = [Z']_{v+1}^{n/2}$. Otherwise he outputs $\bot$.
\end{step}
\end{description}
\end{protocol*}
\caption{The one-round privacy amplification protocol from ~\cite{DKKRS12}.}
\label{fig:pa-one-round}
\end{figure}

\begin{theorem} For any integer $n$ and $k > n/2$, and any $\eps>0$, the protocol in Figure~\ref{fig:pa-one-round} is a one-round $(k,m,\eps)$-quantum secure privacy amplification protocol with post-application robustness and entropy loss $k-m = (n/2)+ \log(1/\eps)$.
\end{theorem}
\begin{proof}
Correctness and extraction follow as in the classical proof by observing that $\Ext(X, Y) = YX_1+X_2$ is a quantum-proof extractor since $h_Y(X_1,X_2) = YX_1+X_2$ is a family of universal hash function, which is shown to be a quantum-proof strong extractor in~\cite{leftoverhash}.
For robustness, the classical proof does not generalize directly. We prove post-application robustness as follows.

We proceed by contradiction.  Suppose post-application robustness is violated, i.e. $\Pr[R_A \neq R_B \land~ R_A \neq \perp \land~ R_B \neq \perp]> \eps$. Then there is an initial state $\rho_{XE}$ with $\Hmin(X|E)_\rho \geq k$ and a CPTP map $T:  \Lin(\mH_Y \otimes \mH_W \otimes \mH_{R_A} \otimes \mH_E) \to \Lin(\mH_Y \otimes \mH_W   \otimes \mH_{E'})$ that can be applied by an adversary Eve to produce a modified message that is accepted by Bob with probability greater than $\eps$. Note that $T$ has $R_A$ as input since we consider post-application robustness. Let $(Y',W',E')=T(Y,W,R_A,E)$. If post-application robustness is violated, then $\Pr[ W'=[Y' X_1+X_2]_1^v] > \eps$. 

Consider the following communication game: Alice has access to a cq-state $\rho_{XE}$. Alice samples a uniformly random $Y$, computes $W=[Y X_1+X_2]_1^v$, $R_A = [Y X_1+X_2]_{v+1}^{n/2}$, and sends $E$, $Y$, $W$, and $R_A$ to Bob. They win if Bob guesses $X$ correctly from $E$, $Y$, $W$, and $R_A$. Using the map $T$ introduced above, Bob can execute the following strategy. First, apply $T$ on Alice's message to generate a guess $(Y',W')$. Second, guess a uniformly random  $R'_B$. Third, use $Y,Y', (W,R_A) = Y X_1+X_2,$ and $(W',R'_B) = Y' X_1+X_2$ to solve for a unique $X=(X_1,X_2)$. Note that Bob succeeds if the guesses $(Y',W')$ and $R'_B$ in the first two steps  are both correct (i.e., $(W',R'_B) = Y' X_1+X_2$), which has probability greater than $\eps \cdot 2^{- ((n/2) - v)}$. On the other hand, we can upper bound the winning probability of the communication game using the min entropy assumption $H(X|E)_\rho \geq k$. Since $Y$ is independent of $X$ and the length of  $(W,R_A)$ is $n/2$, $\Hmin(X|E,Y,W)_\rho \geq k - (n/2)$. Thus the winning probability is less than $2^{-(k-(n/2))}$. Putting the two calculations  together we have 
$$\eps \cdot 2^{- ((n/2) - v)} \leq \Pr[ \mbox{ Bob wins } ] \leq  2^{-(k-(n/2))},$$
which implies $v < n-k-\log(1/\eps)$, a contradiction.
\end{proof}


%
%



\bibliographystyle{alpha}
\bibliography{fuzzy,nmext}

\newcommand{\etalchar}[1]{$^{#1}$}
\begin{thebibliography}{CvDNT99}

\bibitem[ADJ{\etalchar{+}}14]{ADJMR14}
Divesh Aggarwal, Yevgeniy Dodis, Zahra Jafargholi, Eric Miles, and Leonid
  Reyzin.
\newblock Amplifying privacy in privacy amplification.
\newblock In {\em Advances in Cryptology - {CRYPTO} 2014 - 34th Annual
  Cryptology Conference, Santa Barbara, CA, USA, August 17-21, 2014,
  Proceedings, Part {II}}, pages 183--198, 2014.

\bibitem[AHL16]{AHL16}
Divesh Aggarwal, Kaave Hosseini, and Shachar Lovett.
\newblock Affine-malleable extractors, spectrum doubling, and application to
  privacy amplification.
\newblock In {\em Information Theory (ISIT), 2016 IEEE International Symposium
  on}, pages 2913--2917. Ieee, 2016.

\bibitem[BBCM95]{BBCM95}
Charles~H. Bennett, Gilles Brassard, Claude Cr{\'e}peau, and Ueli~M. Maurer.
\newblock Generalized privacy amplification.
\newblock {\em IEEE Transactions on Information Theory}, 41(6):1915--1923,
  1995.

\bibitem[BBR88]{BBR88}
Charles~H. Bennett, Gilles Brassard, and Jean-Marc Robert.
\newblock Privacy amplification by public discussion.
\newblock {\em SIAM Journal on Computing}, 17(2):210--229, 1988.

\bibitem[BF11]{BF11}
Niek~J. Bouman and Serge Fehr.
\newblock Secure authentication from a weak key, without leaking information.
\newblock In {\em Advances in Cryptology - {EUROCRYPT} 2011 - 30th Annual
  International Conference on the Theory and Applications of Cryptographic
  Techniques, Tallinn, Estonia, May 15-19, 2011. Proceedings}, pages 246--265,
  2011.

\bibitem[Bha97]{Bhatia}
Rajendra Bhatia.
\newblock {\em Matrix Analysis}.
\newblock Graduate Texts in Mathematics, Springer, 1997.

\bibitem[CGL15]{CGL15}
Eshan Chattopadhyay, Vipul Goyal, and Xin Li.
\newblock Non-malleable extractors and codes, with their many tampered
  extensions.
\newblock {\em arXiv preprint arXiv:1505.00107}, 2015.

\bibitem[CKOR10]{CKOR10}
Nishanth Chandran, Bhavana Kanukurthi, Rafail Ostrovsky, and Leonid Reyzin.
\newblock Privacy amplification with asymptotically optimal entropy loss.
\newblock In {\em Proceedings of the 42nd {ACM} Symposium on Theory of
  Computing, {STOC} 2010, Cambridge, Massachusetts, USA, 5-8 June 2010}, pages
  785--794, 2010.

\bibitem[CLW14]{clw14}
Kai-Min Chung, Xin Li, and Xiaodi Wu.
\newblock Multi-source randomness extractors against quantum side information,
  and their applications.
\newblock 2014.

\bibitem[Coh15]{Coh15}
Gil Cohen.
\newblock Non-malleable extractors - new tools and improved constructions.
\newblock {\em Electronic Colloquium on Computational Complexity {(ECCC)}},
  22:183, 2015.

\bibitem[CRS12]{CRS12}
Gil Cohen, Ran Raz, and Gil Segev.
\newblock Non-malleable extractors with short seeds and applications to privacy
  amplification.
\newblock In {\em Computational Complexity (CCC), 2012 IEEE 27th Annual
  Conference on}, pages 298--308. IEEE, 2012.

\bibitem[CV16]{CV16}
Gil Cohen and Thomas Vidick.
\newblock Privacy amplification against active quantum adversaries.
\newblock 2016.

\bibitem[CvDNT99]{CDNT97}
Richard Cleve, Wim van Dam, Michael Nielsen, and Alain Tapp.
\newblock Quantum entanglement and the communication complexity of the inner
  product function.
\newblock In {\em Williams C.P. (eds) Quantum Computing and Quantum
  Communications. Lecture Notes in Computer Science}, volume 1509, pages
  61--74. Springer, Berlin, Heidelberg, 1999.

\bibitem[DKK{\etalchar{+}}12]{DKKRS12}
Yevgeniy Dodis, Bhavana Kanukurthi, Jonathan Katz, Leonid Reyzin, and Adam
  Smith.
\newblock Robust fuzzy extractors and authenticated key agreement from close
  secrets.
\newblock {\em {IEEE} Transactions on Information Theory}, 58(9):6207--6222,
  2012.

\bibitem[DLWZ11]{DLWZ11}
Yevgeniy Dodis, Xin Li, Trevor~D. Wooley, and David Zuckerman.
\newblock Privacy amplification and non-malleable extractors via character
  sums.
\newblock In {\em FOCS}, pages 668--677, 2011.

\bibitem[DLWZ14]{DLWZ14}
Yevgeniy Dodis, Xin Li, Trevor~D. Wooley, and David Zuckerman.
\newblock Privacy amplification and nonmalleable extractors via character sums.
\newblock {\em {SIAM} J. Comput.}, 43(2):800--830, 2014.

\bibitem[DP07]{DP07}
Yevgeniy Dodis and Prashant Puniya.
\newblock Feistel networks made public, and applications.
\newblock In Moni Naor, editor, {\em Advances in Cryptology - EUROCRYPT 2007},
  volume 4515 of {\em Lecture Notes in Computer Science}, pages 534--554.
  Spring{\-}er-Ver{\-}lag, 2007.

\bibitem[DPVR12]{DVPR11}
Anindya De, Christopher Portmann, Thomas Vidick, and Renato Renner.
\newblock Trevisan's extractor in the presence of quantum side information.
\newblock 41(4):915--940, 2012.

\bibitem[DW09]{DW09}
Yevgeniy Dodis and Daniel Wichs.
\newblock Non-malleable extractors and symmetric key cryptography from weak
  secrets.
\newblock In Michael Mitzenmacher, editor, {\em Proceedings of the 41st Annual
  ACM Symposium on Theory of Computing}, pages 601--610, Bethesda, MD, USA,
  2009. ACM.

\bibitem[DY13]{DY13}
Yevgeniy Dodis and Yu~Yu.
\newblock Overcoming weak expectations.
\newblock In {\em {TCC}}, pages 1--22, 2013.

\bibitem[GKK{\etalchar{+}}07]{GKKRW07}
Dmitry Gavinsky, Julia Kempe, Iordanis Kerenidis, Ran Raz, and Ronald De~Wolf.
\newblock Exponential separations for one-way quantum communication complexity,
  with applications to cryptography.
\newblock In {\em Proceedings of the thirty-ninth annual ACM symposium on
  Theory of computing}, pages 516--525. ACM, 2007.

\bibitem[KK12]{KK12}
Roy Kasher and Julia Kempe.
\newblock Two-source extractors secure against quantum adversaries.
\newblock {\em Theory of Computing}, 8(1):461--486, 2012.

\bibitem[KRS09]{minentropyguess}
Robert Koenig, Renato Renner, and Christian Schaffner.
\newblock {\em IEEE Transactions on Information Theory}, 55(9), 2009.

\bibitem[Li12a]{Li12b}
Xin Li.
\newblock Design extractors, non-malleable condensers and privacy
  amplification.
\newblock In {\em Proceedings of the 44th Symposium on Theory of Computing
  Conference, {STOC} 2012, New York, NY, USA, May 19 - 22, 2012}, pages
  837--854, 2012.

\bibitem[Li12b]{Li13}
Xin Li.
\newblock Non-malleable condensers for arbitrary min-entropy, and almost
  optimal protocols for privacy amplification.
\newblock {\em CoRR}, abs/1211.0651, 2012.

\bibitem[Li12c]{Li12a}
Xin Li.
\newblock Non-malleable extractors, two-source extractors and privacy
  amplification.
\newblock In {\em FOCS}, pages 688--697, 2012.

\bibitem[Li15]{Li15}
Xin Li.
\newblock Non-malleable condensers for arbitrary min-entropy, and almost
  optimal protocols for privacy amplification.
\newblock In {\em Theory of Cryptography - 12th Theory of Cryptography
  Conference, {TCC} 2015, Warsaw, Poland, March 23-25, 2015, Proceedings, Part
  {I}}, pages 502--531, 2015.

\bibitem[Li17]{Li17}
Xin Li.
\newblock Improved non-malleable extractors, non-malleable codes and
  independent source extractors.
\newblock In {\em Proceedings of the 49th Annual {ACM} {SIGACT} Symposium on
  Theory of Computing, {STOC} 2017, Montreal, QC, Canada, June 19-23, 2017},
  pages 1144--1156, 2017.

\bibitem[LLTT05]{LLST05}
Chia-Jung Lee, Chi-Jen Lu, Shi-Chun Tsai, and Wen-Guey Tzeng.
\newblock Extracting randomness from multiple independent sources.
\newblock {\em IEEE Transactions on Information Theory}, 51(6):2224--2227,
  2005.

\bibitem[Mau92]{Mau92}
Ueli Maurer.
\newblock Conditionally-perfect secrecy and a provably-secure randomized
  cipher.
\newblock {\em Journal of Cryptology}, 5(1):53--66, 1992.

\bibitem[MW97]{MW97}
Ueli Maurer and Stefan Wolf.
\newblock Privacy amplification secure against active adversaries.
\newblock In Burton~S. Kaliski, Jr., editor, {\em Advances in
  Cryptology---CRYPTO~'97}, volume 1294 of {\em LNCS}, pages 307--321.
  Spring{\-}er-Ver{\-}lag, 1997.

\bibitem[NS06]{NS}
Ashwin Nayak and Julia Salzman.
\newblock Limits on the ability of quantum states to convey classical messages.
\newblock {\em Journal of the ACM (JACM)}, 53(1):184--206, 2006.

\bibitem[NZ96]{NZ96}
Noam Nisan and David Zuckerman.
\newblock Randomness is linear in space.
\newblock {\em Journal of Computer and System Sciences}, 52(1):43--53, 1996.

\bibitem[RK05]{renner2005universally}
Renato Renner and Robert K{\"o}nig.
\newblock Universally composable privacy amplification against quantum
  adversaries.
\newblock In {\em Theory of Cryptography}, pages 407--425. Springer, 2005.

\bibitem[RW03]{RW03}
Renato Renner and Stefan Wolf.
\newblock Unconditional authenticity and privacy from an arbitrarily weak
  secret.
\newblock In Dan Boneh, editor, {\em Advances in Cryptology---CRYPTO~2003},
  volume 2729 of {\em LNCS}, pages 78--95. Spring{\-}er-Ver{\-}lag, 2003.

\bibitem[TSSR11]{leftoverhash}
Marco Tomamichel, Christian Schaffner, Adam~D. Smith, and Renato Renner.
\newblock Leftover hashing against quantum side information.
\newblock {\em {IEEE} Trans. Information Theory}, 57(8):5524--5535, 2011.

\bibitem[VDTR13]{vitanov2013chain}
Alexander Vitanov, Frederic Dupuis, Marco Tomamichel, and Renato Renner.
\newblock Chain rules for smooth min-and max-entropies.
\newblock {\em Information Theory, IEEE Transactions on}, 59(5):2603--2612,
  2013.

\end{thebibliography}

\appendix

\section{The Dodis-Wichs Protocol}
\label{app:pa}

In this appendix we reproduce the proof of Theorem~\ref{thm:pa-sec}, taken from~\cite{CV16}. 


\begin{proof}[Proof of Theorem~\ref{thm:pa-sec}] 
Let an \emph{active attack} on Protocol DW be specified by
\begin{itemize}
\item A cq state $\rho_{XE} \in \Density(\mH_X \otimes \mH_E)$ such that $\Hmin(X|E)_\rho \geq k$;
\item A CPTP map $T_1 : \Lin(\mH_Y \otimes \mH_E) \to \Lin(\mH_Y \otimes \mH_{E'})$ whose output on the first registered is systematically decohered in the computational basis; formally, for any $\rho_{YE}$, $T_1(\rho_{YE}) = \sum_{y} (\proj{y}_Y \otimes \Id_E)T_1(\rho_{YE}) (\proj{y}_Y \\ \otimes \Id_E) $; 
\item A CPTP map $T_2 : \Lin(\C^{2^{d_2}} \otimes \C^{2^{t}}\otimes  \mH_{E'}) \to \Lin(\C^{2^{d_2}} \otimes \C^{2^{t}} \otimes \mH_{E''})$.
\end{itemize}
Given an active attack $(\rho_{XE},T_1,T_2)$ we instantiate random variables $Y_A,Z,Y'_A, Y_B,Z',\sigma,Y'_B,\sigma'$ and $R_A,R_B$ in the obvious way, as defined in the protocol and taking into account the maps $T_1$ and $T_2$, applied successively to determine $Y'_A$ and $(Y'_B,\sigma')$.

The correctness of the protocol is clear.

To show robustness, let $\sigma_{ Y'_A Y_A X E'}$ denote the joint state of $Y'_A$, $Y_A$ (which represents a local copy of $Y_A$ kept by Alice), $X$, and Eve's registers after her first map $T_1$ has been applied. Further decompose $\rho$ as a sum of sub-normalized densities $\sigma^=_{Y'_A Y_A X E'}$, corresponding to conditioning on $Y'_A=Y_A$, and $\sigma^\perp_{Y'_A Y_A X E'}$, corresponding to conditioning on $Y'_A\neq Y_A$.

Conditioned on $Y'_A=Y_A$, by definition of a $\mac$ the probability that $(Y'_B,W')\neq (Y_B,W)$ and Alice reaches the $\kc$ state is at most $\eps_{\mac}$. If $(Y'_B,W') = (Y_B,W)$ then $R_A=R_B$, so that in this case robustness holds with error at most $\eps_\mac$.

Now suppose $Y'_A \neq Y_A$. Consider a modified adversary $\Adv'$ that keeps a copy of $Y_A$, applies the map $T_1$, and if $Y'_A=Y_A$ replaces $Y'_A$ with a uniformly random string that is distinct from $Y_A$. This adversary implements a CPTP map $T'_1$ that has no fixed point. By the assumption that $\nmExt$ is a quantum-proof non-malleable extractor,
\begin{equation}\label{eq:pa-sec-1}
\sigma'_{\nmExt(X,Y_A) \nmExt(X,Y'_A) Y_AY'_A E'} \approx_{\eps_\nmExt} U_m \otimes {\sigma'}_{\nmExt(X,Y'_A) Y_A Y'_AE'},
\end{equation}
where here $Y'_AE'$ is defined as the output system of the map $T'_1$ implemented by $\Adv'$. Conditioned on $Y_A \neq Y'_A$ the maps $T_1$ and $T'_1$ are identical, thus it follows from~\eqref{eq:pa-sec-1} and the definition of $\rho^{\perp}$ that
$$\sigma^{\perp}_{\nmExt(X,Y_A) \nmExt(X,Y'_A) Y_AY'_A E'} \approx_{\eps_\nmExt} U_m \otimes \sigma^{\perp}_{\nmExt(X,Y'_A) Y_A Y'_AE'},$$
where now the states are sub-normalized. Since $Z'= \nmExt(X,Y'_A)$ this means that the key used by Alice to verify the signature in Step 3. of Protocol DW is (up to statistical distance $\eps_\nmExt$) uniform and independent of the key used by Bob to make the MAC. By the security of MAC, the probability for Alice to reach the $\kc$ state in this case is at most $\eps_{\nmExt}+\eps_{\mac}$. Adding both parts together, $\Pr(R_A \notin \{R_B,\perp\})\leq \eps_{\nmExt}+2\eps_{\mac}$. Since $R_B$ is never $\bot$, this implies the robustness property. 

For the extraction property, it is sufficient to show that $(R_B,V,E)\approx_{\eps} (U_m,V,E)$ since then key extraction property follows from the robustness and the fact that $R_B$ is never $\bot$. We have that $R_B=\Ext(X,Y_B)$ is close to uniform given $V=Y_AY_BW$ and $E'$, and we need to establish two properties: first, independence between $X$ and $Y_B$ given $Y_AZ'E'$ and second, that the source has enough entropy conditioned on $Y_AZ'E'$.  Regarding the first property, observe that conditioned on $Y_AZ'$, $X$ and $Y_B$ are independent given $E'$. Regarding the source entropy, by the chain rule for the (smooth) min-entropy~\cite{vitanov2013chain}, it follows that
$\Hmin^{\eps_\Ext}(X|Y_AZ'E') \geq k-\log d_Z-c\log(1/\eps_\Ext)$ for some constant $c>0$. Note that 
\begin{align*}
\big\| (R_B,V,E')_\sigma-(U_m,V,E')_\sigma \big\|_1 &\leq \big\| (R_B,Y_A,Y_B,Z',E')_\sigma-(U_m,Y_A,Y_B,Z',E')_\sigma \big\|_1 ,
\end{align*}
which follows since $W$ is a deterministic function $Y_B$ and $Z'$. 
Using that $\Ext$ is a strong quantum-proof extractor, we conclude that $(R_B,V,E)\approx_{\eps} (U_m,V,E)$, as long as $\eps$ is such that $\eps > \eps_\Ext$. 

\end{proof}

\end{document}